\documentclass[11pt]{article}

\usepackage{color}
\definecolor{ForestGreen}{rgb}{0.1333,0.5451,0.1333}
\usepackage[colorlinks=true,linkcolor=ForestGreen,citecolor=ForestGreen]{hyperref}
\usepackage{amsmath, amssymb, amsthm}
\usepackage{mathtools}
\usepackage[noabbrev,capitalize,nameinlink]{cleveref}
\crefname{equation}{}{}
\usepackage{fullpage}
\usepackage{graphics}
\usepackage{pifont}
\usepackage{tikz}
\usepackage{bbm}
\usepackage[T1]{fontenc}

\DeclareMathOperator*{\argmin}{arg\,min}
\usetikzlibrary{arrows.meta}

\usepackage{environ}
\usepackage{framed}
\usepackage{url}
\usepackage[linesnumbered,ruled,vlined]{algorithm2e}
\usepackage[noend]{algpseudocode}
\usepackage[labelfont=bf]{caption}
\usepackage{cite}
\usepackage{framed}
\usepackage[framemethod=tikz]{mdframed}
\usepackage{appendix}
\usepackage{graphicx}
\usepackage[textsize=tiny]{todonotes}
\usepackage{tcolorbox}
\allowdisplaybreaks[1]

\newcommand\remove[1]{}

\newtheorem{theorem}{Theorem}
\newtheorem{lemma}{Lemma}[section]
\newtheorem*{lemma*}{Lemma}
\newtheorem{corollary}[lemma]{Corollary}
\newtheorem*{corollary*}{Corollary}

\theoremstyle{definition}

\newtheorem*{theorem*}{Theorem}
\newtheorem{definition}[lemma]{Definition}

\newtheorem*{rem*}{Remark}

\newtheorem{obs}{Observation}

\newcommand\R{\mathbb{R}}

\newcommand\E{\mathbb{E}}

\newcommand{\eps}{\varepsilon}
\renewcommand{\O}{\widetilde{O}}

\newcommand{\pe}{\preceq}
\newcommand{\se}{\succeq}

\newcommand{\hf}{\widehat{f}}
\newcommand{\assign}{\leftarrow}
\newcommand{\otilde}{\O}
\renewcommand{\forall}{\mathrm{\text{ for all }}}
\renewcommand{\d}{\delta}
\newcommand{\D}{\Delta}

\newcommand{\wpe}{w^+_e}
\newcommand{\wme}{w^-_e}
\newcommand{\npe}{\nu^+_e}
\newcommand{\nme}{\nu^-_e}
\newcommand{\upe}{u^+_e}
\newcommand{\ume}{u^-_e}

\newcommand{\mpe}{\mu^+_e}
\newcommand{\mme}{\mu^-_e}
\newcommand{\spe}{s^+_e}
\newcommand{\sme}{s^-_e}
\newcommand{\cpe}{c^+_e}
\newcommand{\cme}{c^-_e}
\newcommand{\g}{\nabla}

\renewcommand{\E}{D^V}

\newcommand{\new}{\mathrm{new}}

\newcommand{\up}{u^+}
\newcommand{\um}{u^-}
\renewcommand{\wp}{w^+}
\newcommand{\wm}{w^-}

\newcommand{\val}{\mathrm{val}}
\newcommand{\Maxflow}{\textsc{Maxflow}}

\newcommand{\X}{\mathcal{X}}

\newcommand{\ReduceToP}{\textsc{ReduceTo2P}}
\newcommand{\OracleToP}{\textsc{Oracle2P}}
\newcommand{\tE}{\widetilde{D^V_w}}
\newcommand{\tval}{\widetilde{\val}}

\newcommand{\tphi}{\widetilde{\phi}}
\newcommand{\vallp}{\val_{g,r,s}}
\newcommand{\tD}{\widehat{\D}}
\newcommand{\tlog}{\widetilde{\log}}
\newcommand{\tB}{\widetilde{\B}}
\renewcommand{\phi}{D}
\newcommand{\B}{\mathcal{B}}
\newcommand{\tEv}{\widetilde{D^V_{w+\mu}}}
\newcommand{\Augment}{\textsc{Augment}}
\newif\ifrandom
\randomtrue

\newcommand{\maxflowruntime}{m^{4/3+o(1)}U^{1/3}}

\newcommand{\defeq}{\stackrel{\mathrm{\scriptscriptstyle def}}{=}}

\newcommand{\poly}{{\mathrm{poly}}}
\newcommand{\err}{\frac{1}{2^{\poly(\log m)}}}

\newcommand{\todolater}[1]{}

\crefname{algocf}{Algorithm}{Algorithms}

\author{
Yang P. Liu \\
Stanford University \\
\texttt{yangpliu@stanford.edu}
\thanks{Research supported by the U.S.
Department of Defense via an NDSEG fellowship.} 
\and
Aaron Sidford \\
Stanford University \\
\texttt{sidford@stanford.edu}
\thanks{Research supported by NSF CAREER Award CCF-1844855.}
}

\begin{document}

\title{Faster Divergence Maximization for Faster Maximum Flow}

\begin{titlepage}
\clearpage\maketitle
\thispagestyle{empty}

\begin{abstract}
In this paper we provide an algorithm which given any $m$-edge $n$-vertex directed graph with integer capacities at most $U$ computes a maximum $s$-$t$ flow for any vertices $s$ and $t$ in $\maxflowruntime$ time. This improves upon the previous best running times of $m^{11/8+o(1)}U^{1/4}$ (Liu Sidford 2019), $\otilde(m \sqrt{n} \log U)$ (Lee Sidford 2014), and $O(mn)$ (Orlin 2013) when the graph is not too dense or has large capacities.

To achieve the results this paper we build upon previous algorithmic approaches to maximum flow based on interior point methods (IPMs). In particular, we overcome a key bottleneck of previous advances in IPMs for maxflow (M\k{a}dry 2013, M\k{a}dry 2016, Liu Sidford 2019), which make progress by maximizing the energy of local $\ell_2$ norm minimizing electric flows. We generalize this approach and instead maximize the divergence of flows which minimize the Bregman divergence distance with respect to the weighted logarithmic barrier. This allows our algorithm to avoid dependencies on the $\ell_4$ norm that appear in other IPM frameworks (e.g. Cohen M\k{a}dry Sankowski Vladu 2017, Axiotis M\k{a}dry Vladu 2020). Further, we show that smoothed $\ell_2$-$\ell_p$ flows (Kyng, Peng, Sachdeva, Wang 2019), which we previously used to efficiently maximize energy (Liu Sidford 2019), can also be used to efficiently maximize divergence, thereby yielding our desired runtimes. We believe both this generalized view of energy maximization and generalized flow solvers we develop may be of further interest.
\end{abstract}

\end{titlepage}

\newpage

\section{Introduction}
\label{sec:intro}

In this paper, we consider the classic, well-studied, \emph{maxflow} problem of computing a maximum $a$-$b$ flow for vertices $a$ and $b$ in an $m$-edge, $n$-vertex capacitated directed graph with integer capacities at most $U$ (see \cref{sec:prelim} for a precise specification). This problem has been the subject of decades of extensive research, encompasses prominent problems like minimum $a$-$b$ cut and maximum matching in bipartite graphs, and in recent years has served as a prominent proving ground for algorithmic advances in optimization \cite{CKMST11,KLOS14,Sherman13,LS14,Madry13,Madry16,Peng16,ST18,Sherman17}. 

The main result of this paper is a deterministic $\maxflowruntime$ time algorithm for solving maxflow. This runtime improves upon the previous best running times of $m^{11/8+o(1)}U^{1/4}$ \cite{LS19}, $\otilde(m \sqrt{n} \log U)$ \cite{LS14}, and $O(mn)$ \cite{Orlin13} when the graph is not too dense and doesn't have too large capacities.\footnote{Here and throughout the paper we use $\otilde(\cdot)$ to hide $\poly\log(n,m,U)$ factors.} To obtain this result, we provide a new perspective on previous interior point method (IPM) based approaches to $o(m^{3/2})$ runtimes for maxflow \cite{Madry13,Madry16,LS19} and obtain a generalized framework that leads to faster methods (see \cref{sec:overview}). Further, we show how to implement this framework by generalizing previous iterative refinement based algorithms to solve a new class of undirected flow problems to high precision in almost linear time (see \cref{sec:overviewrefine}).

We believe each of these advances is of independent interest and may broadly serve as the basis for faster algorithms for flow problems. We give potential further directions in \cref{sec:conclusion}.

\subsection{History, Context, and Significance}
\label{sec:significance}

\paragraph{Sparse, Unit-capacity Maxflow and Matching.}
To explain the significance of our result, for the remainder of \cref{sec:significance} we restrict our attention to the simplified problem of computing maxflow on unit capacity sparse graphs, i.e. the maxflow problem where $U = 1$ and $m = \O(n)$. This problem is equivalent, up to nearly linear time reductions, to the problems of computing a maximum cardinality set of disjoint $a$-$b$ paths in sparse graph \cite{LRS13} and computing a maximum cardinality matching in a sparse bipartite graph \cite{Lin09}.

Despite the simple nature of this problem, obtaining running time improvements for it is notoriously difficult. In the 1970s it was established that this problem can be solved in $\O(n^{3/2})$ \cite{Karzanov73,ET75} time  and no improvement was found until the 2010s. Even for the easier problem of computing an $\eps$-approximate maxflow, i.e. a feasible flow of value at least $(1-\eps)$ times the optimum, in an undirected graph, no running time improvement was found until the same time period.\footnote{Improvements were known for dense graphs and graphs with bounded maxflow value \cite{Karger97,Karger98a,Karger98b,Karger99}.} 

\paragraph{Energy Maximization Based Breakthroughs.}
This longstanding barrier for algorithmic efficiency was first broken by a breakthrough result of Christiano, Kelner, M\k{a}dry, Spielman and Teng in 2011 \cite{CKMST11}. This work provided an algorithm which computed an $\eps$-approximate maxflow in unit-capacity sparse graphs in time $\O(n^{4/3} \poly(1/\eps))$. Interestingly, this work leveraged a seminal result of Spielman and Teng in 2004 \cite{ST04}, that a broad class of linear systems associated with undirected graphs, known as Laplacian systems, could be solved in nearly linear time. 

Solving Laplacian systems was known to be equivalent to computing (to high precision) \emph{$a$-$b$ electric flows}, the flow $f \in \R^E$ that sends one unit of flow and minimizes \emph{energy}, $\mathcal{E}(f) \defeq \sum_{e \in E} r_e f_e^2$ for input edge resistances $r \in \R^E_{>0}$. \cite{CKMST11} demonstrated that this fact combined with an iterative method, known as the multiplicative weights update (MWU), immediately yields an $\O(n^{3/2} \poly(1/\eps))$-time maxflow algorithm. They then introduced a new technique of \emph{energy maximization} to obtain their improved runtime. More precisely, they noted that when MWU converged slowly, a small number of edges could be removed to dramatically increase the energy of the electric flows considered. By trading off the loss in approximation quality from edge removal with possible energy increase, this paper achieved its breakthrough $\O(n^{4/3} \poly(1/\eps))$ runtime.

This result also immediately created a tantalizing opportunity to obtain comparable improvements for solving the problem exactly on directed graphs. In 2008, Daitch and Spielman \cite{DS08} had already shown that another powerful class of continuous optimization methods, IPMs, reduce the problem of solving maxflow exactly on directed graphs to solving $\O(n^{1/2})$ Laplacian systems. Consequently, this raised a natural question of whether energy maximization techniques of \cite{CKMST11} could be combined with IPMs to achieve faster running times.

In 2013, M\k{a}dry \cite{Madry13} provided a breakthrough result that this was indeed possible and obtained an $\O(n^{10/7}) = \O(n^{3/2 - 1/14})$ time algorithm for directed maxflow on sparse unit capacity graphs. This result was an impressive tour de force that involved multiple careful modifications to standard IPM analysis. Leveraging energy maximization techniques is intrinsically more difficult in IPMs than in MWU, where there is a type of monotonicity that does not occur naturally in IPMs. Additionally, several aspects of IPMs are somewhat brittle and tailored to $\ell_2$ and $\ell_4$ norms, rather than $\ell_\infty$ as in maxflow. Consequently, \cite{Madry13} performed multiple modifications to carefully track how both energy and IPM invariants changed during the IPM. While this sufficed to achieve the first $n^{3/2 - \Omega(1)}$ time maxflow algorithm, the aforementioned difficulties caused the ultimate running time to be slower than the $\O(n^{4/3})$ runtime for approximate maxflow on undirected graphs.

\vspace{-5 pt}
\paragraph{Beyond Energy Maximization.} Since \cite{CKMST11}, energy maximization has been applied to a host of other problems \cite{KMP12,CMMP13,AKPS19}. Further, \cite{Madry13} has been simplified and improved \cite{Madry16}, and applied to related graph problems \cite{CMSV16}. Very recently the runtime was improved by the authors to $n^{11/8+o(1)} = n^{3/2-1/8+o(1)}$ \cite{LS19} and this in turn lead to faster algorithms for mincost flow and shortest paths with negative edge lengths \cite{AMV20}. These works address IPM energy monotonicity issues in a novel way, but do not run in time $m^{4/3}$ due to issues of maintaining IPM invariants and working with $\ell_4$, rather than $\ell_\infty$ (see \cref{sec:overview}).\footnote{Technically, in \cite{LS19} and \cite{AMV20}, weight changes are computed to reduce the $\ell_\infty$ norm of congestion of an electric flow vector. However, the centrality depends on the $\ell_4$ norm, and we see this as why previous works achieved worse than a $m^{4/3}$ runtime. Since the initial version of this paper was released, \cite{AMV20} was updated to leverage the techniques of this paper and achieved a $m^{4/3+o(1)}$ runtime for a broader range of problems.}

In light of this past decade of advances to maxflow, this paper makes several contributions. First, we obtain an $n^{4/3 + o(1)}$ maxflow algorithm on sparse, unit-capacity graphs, thereby closing the longstanding gap between the runtime achieved for approximately solving maxflow using Laplacian system solvers and the best runtime known for solving maxflow on directed graphs.\footnote{Since \cite{CKMST11}, faster running times for approximate maxflow on undirected graphs have been achieved \cite{LRS13,Sherman13,KLOS14,Peng16,Sherman17,ST18} and the problem is now solvable in nearly linear time. It is unclear whether these improvements should translate to faster directed maxflow runtimes, though this paper and \cite{LS19} both leverage results from a related line of work on solving undirected $\ell_p$-flow problems to high precision \cite{AKPS19,KPSW19,AS20}.}

Second, we shed light on the energy maximization framework that underlay all previous $n^{3/2-\Omega(1)}$ results for exact directed maxflow and  depart from it to achieve our bounds. Energy maximization, though straightforward to analyze, is somewhat mysterious from an optimization perspective -- it is unclear what (if any) standard optimization technique it corresponds to.\footnote{The recent, simultaneous paper of \cite{AMV20} gives an alternative perspective on energy maximization as regularized Newton steps.} In this paper, we show that energy maximization arises naturally when locally optimizing the second order Taylor approximation of a Bregman divergence with respect to a logarithmic barrier. We then show that by optimizing the entire function, instead of its second order Taylor approximation, we can obtain improved convergence. We believe this  \emph{divergence maximization technique} is of intrinsic interest.

Finally, to achieve our result, we show that divergence maximization can be performed efficiently for graphs. Whereas in prior work we showed that energy maximization could be performed efficiently by solving smoothed $\ell_2$-$\ell_p$ flows of \cite{KPSW19,AS20}, here we need to solve problems not immediately representable this way. Nevertheless, we show previous solvers can be applied to solve a quadratic extension of the divergence maximization problem, which suffices for our algorithms. More generally, we open up the algorithms which underlie $\ell_2$-$\ell_p$ flow solvers and show that a range of undirected flow optimization problems, including divergence maximization, can be solved efficiently. We hope this generalized flow solving routine may find further use as well.

\subsection{Our Results}
\label{sec:results}

The main result of this paper is the following theorem.

\begin{theorem}[Maximum Flow] \label{thm:main}
There is an algorithm which solves the maximum flow problem in $m$-edge, integer capacitated graphs with maximum capacity $U$ in $\maxflowruntime$ time.
\end{theorem}

\cref{thm:main} yields an exact maxflow runtime matching the $\O(mn^{1/3}\eps^{-11/3})$ runtime of \cite{CKMST11} for $(1-\eps)$-approximate undirected maxflow on sparse graphs. Further, this improves on the recent $m^{11/8+o(1)}U^{1/4}$ time algorithm of the authors \cite{LS19} as long as $U \le m^{1/2-\eps}$ for some $\eps > 0$. When $U \ge m^{1/2}$, the result of \cref{thm:main} and all the algorithms of \cite{Madry13,Madry16,LS19} have runtime $\O(m^{3/2})$, which is already known through \cite{GR98}. Hence, we assume $U \le \sqrt{m}$ throughout the paper. An immediate corollary of \cref{thm:main} is the following result on bipartite matching.
\begin{corollary}[Bipartite Matching] \label{cor:matching}
There is an algorithm which given a bipartite graph with $m$ edges computes a maximum matching in time $m^{4/3+o(1)}.$
\end{corollary}

We note that \cref{thm:main,cor:matching} are deterministic. This is due to the application \cite{CGLNPS19}, which recently showed that algorithms for many flow algorithms, including Laplacian system solvers \cite{ST04}, smoothed $\ell_2$-$\ell_p$ flow computation methods \cite{KPSW19,AS20} and some maxflow IPMs \cite{Madry13,Madry16}, may be derandomized with a $m^{o(1)}$ runtime increase.

Additionally, we show in \cref{sec:opt} that flow problems that are combinations a quadratic and $\ell_p$ norm of functions with stable Hessians may be solved to high accuracy in almost linear time. This encompasses problems such as computing electric and $\ell_p$ norm minimizing flows, smoothed $\ell_2$-$\ell_p$ flows, and the divergence maximizing flows our algorithm must compute.

\begin{theorem}
\label{thm:mainopt}
For graph $G=(V,E)$ and all $e \in E$, let $0 \le a_e \le 2^{\poly(\log m)}$ be constants, $q_e:\R\to\R$ be functions with $|q_e(0)|, |q_e'(0)| \le 2^{\poly(\log m)}$ and $a_e/4 \le q_e''(x) \le 4a_e$ for all $x\in\R$, and $h_e:\R\to\R$ be functions with $h_e(0) = h_e'(0) = 0$ and $1/4 \le h_e''(x) \le 4$ for all $x\in\R.$ For demand $d$ with entries bounded by $2^{\poly(\log m)}$, even integer $p \in (\omega(1), (\log n)^{2/3-o(1)})$, and all flows $f$ define
\[ \val(f) \defeq \sum_{e\in E} q_e(f_e) + \left(\sum_{e\in E} h_e(f_e)^p\right)^{1/p} \enspace \text{ and } \enspace OPT \defeq \min_{B^Tf=d} \val(f). \] We can compute in time $m^{1+o(1)}$ a flow $f'$ with $B^Tf'=d$ and $\val(f') \le OPT + \err.$
\end{theorem}

\subsection{Comparison to Previous Work}
\label{sec:previouswork}

For brevity, we cover the lines of research most relevant to our work, including IPM based maxflow algorithms, and smoothed $\ell_2$-$\ell_p$ flow algorithms. (See Section 1.3 of \cite{LS19} for further references.)

\paragraph{IPMs for Maxflow:} In directed sparse graphs, the only improvements over the classic \\ $\O( \min \{ m^{3/2}, m n^{2/3}  \} )$ running time of Karzanov and Even-Tarjan \cite{Karzanov73,ET75} are based on IPMs. These are the $\O(m\sqrt{n})$ time algorithm of \cite{LS14} and the $m^{11/8+o(1)}$ algorithm of \cite{LS19}, which improved on the breakthrough $\O(m^{10/7})$ algorithm of \cite{Madry13,Madry16}. Our major departure from these prior works is to analyze the quality of weight changes through divergence distances rather than electric energy, and then directly use the resulting flows. Simultaneously and independent of our work, \cite{AMV20} built on \cite{LS19} and showed that one can directly augment with smoothed $\ell_2$-$\ell_p$ flows. However, their running time was still $m^{11/8+o(1)}$ as these flows resulted from considering energy maximization and regularizing electric flow objectives. We hope that our methods open the door to the design of new optimization methods for solving maxflow beyond interior point methods, and the use of other undirected flow primitives other than electric flows such as smoothed $\ell_2$-$\ell_p$ flows to solve exact maxflow.

\paragraph{Iterative Refinement and Smoothed $\ell_2$-$\ell_p$ Flows:} Our improvements leverage advances in a recent line of work on obtaining faster runtimes solving $\ell_p$-regression to high precision \cite{BCLL18,AKPS19,APS19,AS20}. In particular, as was the case in \cite{LS19}, we leverage recent work of \cite{KPSW19} which showed that for any graph $G = (V,E)$ with incidence matrix $B$, vector $g\in \R^E$, resistances $r \in \R_{\ge0}^E$, demand vector $d$, and $p \in \left[\omega(1), o((\log n)^{2/3})\right]$, there is an algorithm that computes a high accuracy solution to $\min_{B^Tf=d} g^Tf + \sum_e r_ef_e^2 + \|f\|_p^p$ in $m^{1+o(1)}$ time, which we call a \emph{smoothed $\ell_2$-$\ell_p$ flow}. \cite{KPSW19} achieved this result from the \emph{iterative refinement} framework \cite{AKPS19}, which says that one can solve a smoothed $\ell_2$-$\ell_p$ flow to high accuracy through solving $\O(2^{O(p)})$ other smoothed $\ell_2$-$\ell_p$ flow instances approximately. The $m^{11/8+o(1)}$ algorithm of \cite{LS19} used this primitive to compute weight changes to reduce the congestion of the electric flow to make progress. In this work, in \cref{sec:overviewrefine,sec:efficient,sec:refine},  we open up the iterative refinement framework of \cite{AKPS19,KPSW19} to solve the more complex problems that arise from working with objectives beyond electric flows. 

\subsection{Paper Organization}
\label{sec:organization}
In \cref{sec:prelim} we give the preliminaries. In \cref{sec:overview} we give a high level overview of our algorithm, in \cref{sec:detail} describe various pieces in more detail, and our main algorithm is presented in \cref{sec:energymax}.

Missing proofs are deferred to \cref{sec:proofs}, and necessary lemmas for iterative refinement of our objectives are given in \cref{sec:refine}. In \cref{sec:optprelim}  we give additional convex optimization preliminaries and in \cref{sec:opt} we prove \cref{thm:mainopt}, which shows that a large class of flow problems on graphs may be efficiently solved by reduction to smoothed $\ell_2$-$\ell_p$ flows.
\section{Preliminaries}
\label{sec:prelim}

\paragraph{General notation}
We let $\R_{\ge\alpha}^m$ denote the set of $m$-dimensional real vectors which are entrywise at least $\alpha$. For $v \in \R^m$ and real $p \ge 1$ we define $\|v\|_p$, the $\ell_p$-norm of $v$, as $\|v\|_p \defeq \left(\sum_{i=1}^m |v_i|^p \right)^{1/p}$, and  $\|v\|_\infty \defeq \max_{i=1}^m |v_i|.$ For $n \times n$ positive semidefinite matrices $M_1, M_2$ we write $M_1 \approx_r M_2$ for $r \ge 1$ if $r^{-1}x^TM_1x \le x^TM_2x \le rx^TM_1x$ for all $x \in \R^n$. For a differentiable function $f:\R^n\to\R$ we define its induced \emph{Bregman divergence} as $D_f(x\|y) \defeq f(x)-f(y)-\g f(y)^T(y-x)$ for all $x,y \in \R^n$.

\paragraph{Graphs} Throughout this paper, in the graph problems we consider, we suppose that there are both upper and lower capacity bounds on all edges. We let $G$ be a graph with vertex set $V$, edge set $E$, and upper and lower capacities $\up_e \ge 0$ and $\um_e \ge 0$ respectively on edge $e$. We use $U$ to denote the maximum capacity of any edge, so that $\max\{\up_e, \um_e\} \le U$ for all edges $e$. We let $n$ denote the number of vertices $|V|$, and let $m$ denote the number of edges $|E|$. Further we view undirected graphs as directed graphs with $\upe = \ume$ by arbitrarily orienting its edges.

\paragraph{The Maximum Flow Problem} 
Given a graph $G = (V,E)$ we call any assignment of real values to the edges of $E$, i.e. $f \in \R^E$, a \emph{flow}. For a flow $f \in \R^E$, we view $f_e$ as the amount of flow on edge $e$. If $f_e > 0$ we interpret this as sending $f_e$ units in the direction of the edge orientation and if $f_e < 0$ we interpret this as sending $|f_e|$ units in the direction opposite the edge orientation. 

In this paper we consider $ab$-flows, where $a \in V$ is called the  source, and $b \in V$ is called the sink. An $ab$-flow is a flow which routes $t$ units of flow from $a$ to $b$ for some real number $t \ge 0$. Define the unit demand vector $\chi_{ab} = 1_b - 1_a$, a vector with a $1$ in position $a$ and $-1$ in position $b$. When $a$ and $b$ are implicit, we write $\chi = \chi_{ab}$. In this way, we also refer to an $ab$-flow which routes $t$ units from $a$ to $b$ as a $t\chi$-flow. The \emph{incidence matrix} for a graph $G$ is an $m \times n$ matrix $B$, where the row corresponding to edge $e = (u,v)$ has a $1$ (respectively $-1$) in the column corresponding to $v$ (respectively $u$). Note that $f \in \R^E$ is a $t\chi$-flow if and only if $B^Tf=t\chi$. More broadly, we call any $d \in \R^V$ a demand vector if $d \perp 1$ and we say $f \in \R^E$ routes $d$ if $B^T f = d$. 

We say that a $t\chi$-flow $f$ is \emph{feasible} in $G$ if $-\ume \le f_e \le \upe \forall e \in E,$ so that $f$ satisfies the capacity constraints. We define the \emph{maximum flow problem} as the problem of given a graph $G$ with upper and lower capacities $\up$ and $\um$, and source and sink vertices $a$ and $b$, to compute a maximum feasible $ab$-flow. We denote the maximum value as $t^*$. For a real number $t \le t^*$, we let $F_t \defeq t^*-t$ denote the remaining amount of flow to be routed.
\section{Algorithm Derivation and Motivation}
\label{sec:overview}

In this section we present a principled approach for deriving our new method and the previous energy-based methods \cite{Madry13,Madry16,LS19} for maxflow from an IPM setup. 

\subsection{Interior Point Method Setup}

The starting point for our method is the broad IPM framework for maxflow of \cite{LS19}, which in turn was broadly inspired by \cite{Madry16}. We consider the setup described in \cref{sec:prelim} and design algorithms that maintain a flow $f \in \R^E$, a parameter $t \ge 0$, and weights $\wp, \wm \in \R_{\ge1}^m$ such that $B^T f = t \chi$ and $f$ is a high accuracy approximation of the solution $f_{t,w}^*$ of the following optimization problem:
\begin{equation}
\label{eq:central_path_formula}
f_{t,w}^* \defeq \argmin_{B^Tf=t\chi} V(f)
\enspace
\text{ where }
\enspace
V(f) \defeq -\sum_{e\in E} \left(\wp_e \log(\up_e-f_e) + \wm_e \log(\um_e+f_e)\right)
~.
\end{equation}
Here, $V(f)$ is known as the \emph{weighted logarithmic barrier} and penalizes how close $f$ is to violating the capacity constraints and $t$ is the amount of flow sent from $a$ to $b$.

Generally, IPMs proceed towards optimal solutions by iteratively improving the quality, i.e. increasing the parameter $t$, and decreasing the proximity to the constraints, i.e. decreasing $V(f)$. Previous maxflow IPMs \cite{LS14,Madry13,Madry16,LS19} all follow this template. Specifically, \cite{Madry16,LS19} alternate between Newton step to improve the optimality of $f$ for \cref{eq:central_path_formula} (called \emph{centering steps}) and computing a new flow and weights to approximately solve \cref{eq:central_path_formula} for a larger value of $t$ (called \emph{progress steps}). Applying such an approach, while using Laplacian system solvers to implement the steps in nearly linear time, coupled with a preliminary graph preprocessing step known as preconditioning (\cref{sec:precondition}) directly yields to an $\O(m^{3/2})$ time algorithm. Recent advances \cite{LS14,Madry13,Madry16,LS19} were  achieved by performing further work to modify the weights and flows used. 

\subsection{Progress steps via divergence minimization}
\label{sec:divminoverview}

To understand (and improve upon) previous maxflow IPMs, here we  explain how to view progress steps in this framework as computing a divergence minimizing $\d\chi$-flow. Note that, without weight changes, the cumulative result of a progress and centering step is essentially moving from $f_{t,w}^*$ to $f_{t+\d,w}^*$ for a step size $\d$. The optimality conditions of \cref{eq:central_path_formula} give that the gradient of $V$ at the optimum $f^*_{t,w}$ of \cref{eq:central_path_formula} is perpendicular to the kernel of $B^T$, so there is a vector $y$ with $By = \g V(f_{t,w}^*)$. Define
\begin{equation} \label{eq:divminflow} \hf \defeq \argmin_{B^Tf=\d\chi} D_V(f_{t,w}^*+f \| f_{t,w}^*) = \argmin_{B^Tf=\d\chi} V(f_{t,w}^*+f) - V(f_{t,w}^*) - \g V(f_{t,w}^*)^T f, \end{equation} i.e. the $\d\chi$-flow with smallest divergence from $f_{t,w}^*$ against the barrier $V$. Again, optimality conditions give that there is a vector $z$ with $Bz = \g D_V(f_{t,w}^*+\hf \| f_{t,w}^*) = \g V(f_{t,w}^*+\hf) - \g V(f_{t,w}^*).$ Therefore, $B(y+z) = \g V(f_{t,w}^*+\hf)$. Since $f_{t,w}^*+\hf$ is a $(t+\d)\chi$-flow, we must have $f_{t,w}^*+\hf = f_{t+\d,w}^*$ by optimality conditions, so that adding $\hf$ to an optimal point $f_{t,w}^*$ lands us at the next point $f_{t+\d,w}^*$.

Now, a standard progress step in this framework may be computed by taking a Newton step, i.e. minimizing the second order Taylor approximation of the divergence. The second order Taylor expansion of $D_V(f_{t,w}^*+f \| f_{t,w}^*)$ is $\frac12 f^T\g^2 V(f_{t,w}^*)f$, and the resulting step is
\begin{equation} \argmin_{B^Tf=\d\chi} \frac12 f^T\g^2 V(f_{t,w}^*)f = \d \g^2 V(f_{t,w}^*)^{-1}B(B^T\g^2 V(f_{t,w}^*)^{-1}B)^\dagger \chi 
~.
\label{eq:prev} \end{equation} This can be computed in $O(m)$ time plus the time to solve a Laplacian system, i.e. $\O(m)$ \cite{ST04}. Choosing $\d$ that routes $\Omega(m^{-1/2})$ fraction of the remaining flow, adding the flow in \cref{eq:prev} to our current point, and taking further Newton steps to re-center yields an $\O(m^{3/2})$ time maxflow algorithm.

\subsection{Energy-based improvements}

Improvements to the above $\O(m^{3/2})$ time algorithm \cite{Madry13,Madry16,LS19} arise by a more careful analysis of the step size $\d$ of the Newton step that we may take such that recentering may still be performed in $\O(m)$ time, and by leveraging that the flow in \cref{eq:prev} is an electric flow. Precisely, the size of the step we may take is governed by the \emph{congestion} of the flow we wish to add, which is defined edge-wise as the ratio of flow on an edge to its residual capacity (see $\cpe,\cme$ in \cref{sec:setup}). In this way, the $\ell_\infty$ norm of congestion of the $\chi$-electric flow governs the amount of flow we may add before violating capacity constraints. On the other hand, because the $\chi$-electric flow was a minimizer to a second order approximation of the divergence, the $\ell_4$ norm of congestion of the $\chi$-electric flow instead governs the amount of flow we may add so that centering can still be performed in $\O(m)$ time, whereas a bound on the $\ell_2$ norm of congestion suffices to achieve the $\O(m^{3/2})$ time algorithm.

In this way, it is natural to attempt to compute weight changes that reduce the $\ell_4$ norm of congestion induced by the $\chi$-electric flow. M\k{a}dry \cite{Madry13,Madry16} achieves this by boosting edges with high congestion in the $\chi$-electric flow and trading off against a potential function that is the energy of the $\chi$-electric flow with resistances induced by the Hessian of the weighted logarithmic barrier at the current point.

To improve on the algorithm of \cite{Madry16}, \cite{LS19} instead views increasing energy via budgeted weight change as its own optimization problem. Precisely, the optimization problem was to maximize the energy of an electric flow in a graph $G$ that originally had resistances $r$ under a resistance increase budget. Written algebraically, for a weight budget $W$, this is
\begin{equation} 
\max_{\|r'\|_1 \le W} \min_{B^Tf=\d\chi} \sum_{e\in E} (r_e+r_e')f_e^2. \label{eq:optproblem} \end{equation}
\cite{LS19} showed that a smoothed version of this objective was solvable in $m^{1+o(1)}$ time using smoothed $\ell_2$-$\ell_p$ flows \cite{KPSW19}, and that the combinatorial edge boosting framework of \cite{Madry13,Madry16} can essentially be viewed as greedily taking gradient steps against the objective in \cref{eq:optproblem}.

\subsection{Our new method: beyond electric energy}

A disappointing feature of the above discussion is that while the $\ell_\infty$ norm of congestion governs the amount of flow we may add and still have a feasible flow, we are forced to control the $\ell_4$ norm of congestion of the electric flow to allow for efficient centering. This is due to the fact that although the step can be taken without violating capacity constraints, there is sufficient loss in local optimality that $\O(1)$ centering Newton steps cannot recenter it. This leads to the heart of our improvement -- we resolve this discrepancy between the $\ell_\infty$ and $\ell_4$ norm of congestion by directly augmenting via the divergence minimizing flow of \cref{eq:divminflow}. As a result, it suffices to compute weight changes to minimize the $\ell_\infty$ norm of congestion of the divergence minimizing flow.

The next challenge is to compute this divergence minimizing flow, and compute weight changes to reduce the $\ell_\infty$ norm of its congestion. To approach this, we consider the problem of moving from $f_{t,w}^*$ to $f_{t+\d,w}^*$ for a step size $\d$, assuming that the residual capacities induced by $f_{t,w}^*$ and $f_{t+\d,w}^*$ are within $1.05$ multiplicatively. This implies that $\g^2 V(f_{t,w}^*) \approx_{1.2} \g^2 V(f_{t+\d,w}^*)$. To solve this problem, for each piece of the $V(f)$ objective, i.e. $\left(\wp_e \log(\up_e-f_e) + \wm_e \log(\um_e+f_e)\right)$, we replace it with a \emph{quadratic extension}, a function that agrees with it on some interval, and extends quadratically outside. Our new objective will have a stable Hessian everywhere, hence can be minimized by Newton's method. By construction, the optimum of the quadratically extended problem and original are the same using convexity (see \cref{obs:convex}). Further details are provided in \cref{sec:quadsmooth}.

Finally, we must compute weights that reduce the $\ell_\infty$ norm of congestion of the divergence minimizing flow. As the approach of \cite{LS19} computes weight changes to maximize the electric energy, we instead compute weight changes to maximize the divergence of the divergence minimizing flow. Doing this requires extending the analysis of \cite{Madry16} and energy maximization framework of \cite{LS19} to nonlinear objectives, such as the quadratic extensions described above, and then generalizing the iterative refinement framework introduced by \cite{AKPS19, KPSW19} to a large family of new objectives. We hope that both this unified view of energy and divergence maximization as well as the methods we give for performing this optimization efficiently may have further utility.

\section{Technical Ingredients}
\label{sec:detail}
In this section, we elaborate on several technical aspects discussed in \cref{sec:overview}. We give details for setting up the IPM in \cref{sec:setup}, discuss preconditioning in \cref{sec:precondition}, elaborate on quadratic extensions in \cref{sec:quadsmooth}, and discuss iterative refinement in \cref{sec:overviewrefine}.

\subsection{IPM Details and Preconditioning}
\label{sec:setup}
\label{sec:centralpath}
\label{sec:precondition}

In this section, we give a detailed description of our IPM setup. One can reduce directed maxflow to undirected maxflow with linear time overhead (see  \cite{Lin09,Madry13} or \cite{LS19} Section B.4) and consequently, we assume our graph is undirected, so that $\upe = \ume$.

Assuming that there is a feasible $t\chi$-flow, optimality conditions of \cref{eq:central_path_formula} give that the gradient of $V$ at the optimum $f^*_{t,w}$ of \cref{eq:central_path_formula} is perpendicular to the kernel of $B^T$, i.e. there is a dual vector $y \in \R^V$ such that
$B y = \g V(f^*_{t,w})$. Consequently, for parameter $t$ and weight vectors $\wp, \wm$ we say that a flow $f$ is on the \emph{weighted central path} if and only if there exists a dual vector $y \in \R^V$ such that
\begin{align}
B^Tf
= t\chi
\enspace
\text{ and }
\enspace
[By]_e
=  [\g V(f)]_e = \frac{\wp_e}{\up_e-f_e} - \frac{\wm_e}{\um_e+f_e}
\text{ for all } e\in E
 \label{eq:centralpath}
\end{align}
For simplicity, we write $w = (w^+, w^-) \in \R^{2E}_{\ge1}$, where we define $\R^{2E}_{\ge\alpha} \defeq \R^E_{\ge \alpha} \times \R^E_{\ge \alpha}.$
We define \emph{residual capacities} $\cpe \defeq \upe-f_e, \cme \defeq \ume+f_e$ and $c_e = \min(\cpe, \cme).$ Note $c_e \ge 0$ for all $e \in E$ if and only if $f$ is feasible.

We initialize $\wpe = \wme = 1$, $t = 0$, and $f = 0$, which is central. Previous IPM based algorithms for maxflow \cite{Madry13,Madry16,LS19} alternated between progress steps and centering steps. Progress steps increase the path parameter $t$ at the cost of centrality, which refers to the distance of $f$ from satisfying \cref{eq:centralpath} in the inverse norm of the Hessian of $V(f)$ -- see \cite{LS19} Definition 4.1. Centering steps improve the centrality of the current point without increasing the path parameter $t$. Our algorithm more directly advances along the central path -- given parameter $t$, weights $w$, and central path point $f_{t,w}^*$ we compute new weights $w^\new$, advance the path parameter to $t+\d$, and compute $f_{t+\d,w^\new}^*$.

The goal of the IPM is to reduce the value of the residual flow $F_t = t^*-t$ below a threshold, at which point we may round and use augmenting paths \cite{LRS13}. We assume that the algorithm knows $F_t$ throughout, as our algorithm succeeds with any underestimate of the optimal flow value $t^*$, and we can binary search for the optimal flow value.

\paragraph{Preconditioning.} To precondition our undirected graph $G$, we add $m$ undirected edges of capacity $2U$ between source $a$ and sink $b$. This increases the maximum flow value by $2mU.$ Throughout the remainder of the paper, we say that the graph $G$ is preconditioned if it is undirected and we have added these edges. Intuitively, preconditioning guarantees that a constant fraction of the remaining flow in the residual graph may be routed in its undirectification, i.e. $G$ with capacities $c_e$. The following lemma was shown in \cite{LS19} Section B.5.
\begin{lemma}
\label{lemma:precon}
Consider a preconditioned graph $G$. For parameter $t$ and weights $w$ let $c_e$ be the residual capacities induced by the flow $f_{t,w}^*.$ Then for every preconditioning edge $e$ we have that $c_e \ge \frac{F_t}{7\|w\|_1}.$ If $\|w\|_1 \le 3m$ then $c_e \ge \frac{F_t}{21m}$.
\end{lemma}
At the start of the algorithm, as we initialized $\wp = \wm = 1$, we have $\|w\|_1 = 2m.$ To apply \cref{lemma:precon} we maintain the stronger invariant that $\|w\|_1 \le 5m/2$ before each step, but may temporarily increase to $\|w\|_1 \le 3m$ during the step.

\subsection{Advancing along the central path via quadratic smoothing}
\label{sec:quadsmooth}
Let $t$ be a path parameter, and let $\d$ be a step size. Let $\cpe, \cme$ be the residual capacities induced by $f_{t,w}^*$, and let $(\cpe)',(\cme)'$ be those induced by $f_{t+\d,w}^*.$ We sketch an algorithm that computes $f_{t+\d,w}^*$ to high accuracy from $f_{t,w}^*$ in $\O(m)$ time given that for all $e \in E$ that $\cpe \approx_{1.05} (\cpe)'$ and $\cme \approx_{1.05} (\cme)'$.
Let $\hf = f_{t+\d,w}^* - f_{t,w}^*$ and define the change in the value of the barrier $V$ when we add $f$ as
\begin{align} \B(f) &\defeq V(f+f_{t,w}^*)-V(f_{t,w}^*) \nonumber \\ &= -\sum_{e\in E} \left(\wpe\log\left(1-\frac{f_e}{\upe-[f_{t,w}^*]_e}\right)+\wme\log\left(1+\frac{f_e}{\ume+[f_{t,w}^*]_e}\right) \right), \label{eq:important}
\end{align}
so that $\hf = \argmin_{B^Tf=\d\chi} \B(f)$.
To compute $\hf$, we smooth \cref{eq:important} by replacing each instance of $\log(\cdot)$ with a function $\tlog(\cdot)$ defined as
\[ \label{eq:tlog} \tlog_\eps(1+x) \defeq \begin{cases}
    \log(1+x) & \text{ for } |x| \le \eps\\
    \log(1+\eps)+\log'(1+\eps)(x-\eps)+\frac{\log''(1+\eps)}{2}(x-\eps)^2 & \text{ for } x \ge \eps \\
    \log(1-\eps)+\log'(1-\eps)(x+\eps)+\frac{\log''(1-\eps)}{2}(x+\eps)^2 & \text{ for } x \le -\eps.
    \end{cases} \]
Here, we fix $\eps = 1/10$ and write $\tlog(1+x) \defeq \tlog_{1/10}(1+x)$. Note that $\tlog_\eps(1+x)$ is the natural quadratic extension of $\log(1+x)$ outside the interval $|x|\leq \epsilon$. Specifically, the functions agree for $|x| \le \eps$, and we $\tlog_\eps(1+x)$ is the second order Taylor expansion of $\log(1+x)$ at $\eps, -\eps$ for $x > \eps$, $x < -\eps$ respectively. In this way, $\tlog(1+x)$ is twice differentiable everywhere. Define
\[ \tB(f) \defeq -\sum_{e\in E} \left(\wpe\tlog\left(1-\frac{f_e}{\upe-[f_{t,w}^*]_e}\right)+\wme\tlog\left(1+\frac{f_e}{\ume+[f_{t,w}^*]_e}\right) \right). \]

We now claim that $\hf = \argmin_{B^Tf=\d\chi} \tB(f)$, and that it can be computed in $\O(m)$ time. To argue the latter, note that by construction, all Hessians $\g^2 \tB(f)$ are within a multiplicative factor of $2$ of each other, hence we can compute $\argmin_{B^Tf=\d\chi} \tB(f)$ in $\O(m)$ time using Newton's method and electric flow computations. Because $\cpe \approx_{1.05} (\cpe)'$ and $\cme \approx_{1.05} (\cme)'$, we know that $\B$ and $\tB$ agree in a neighborhood of $\hf$, so $\hf = \argmin_{B^Tf=\d\chi} \tB(f)$ by the following simple observation. For completeness, we provide a proof in \cref{proofs:obsconvex}.
\begin{obs}
\label{obs:convex}
Let $\chi \subseteq \R^n$ be a convex set, and let $f, g:\R^n \to \R$ be convex functions. Let $x^* = \argmin_{x\in \chi} f(x)$, and assume that $f, g$ agree on a neighborhood of $x^*$ in $\R^n$. Then $g(x^*) = \min_{x \in \chi} g(x)$.
\end{obs}
We emphasize that we do not directly do the procedure described here, and instead apply quadratic smoothing in a different form in \cref{sec:energymax}. There we smooth the function \\ $\phi \defeq D_{-\log(1-x)}(x\|0)$, the Bregman divergence of $x$ to $0$ with respect to the function $-\log(1-x)$, instead of directly smoothing $\log(1+x).$ The smoothed function $\tphi$ is shown in \cref{eq:defphi,eq:deftphi}.

\subsection{Iterative refinement}
\label{sec:overviewrefine}

The idea of \emph{iterative refinement} was introduced in \cite{AKPS19,KPSW19} to solve $p$-norm regression problems, such as $\min_{B^Tf=d} \|f\|_p^p.$ Iterative refinement solves such an objective by reducing it to approximately minimizing objectives which are combinations of quadratic and $\ell_p$ norm pieces. Specifically, we can show using \cref{lemma:kpswiter} that for a fixed flow $f$ there is a gradient vector $g$ such that for any circulation $\D$ we have
\begin{align} \|f+\D\|_p^p - \|f\|_p^p = g^T\D + \Theta_p\left(\sum_{e\in E} |f_e|^{p-2}\D_e^2 + \|\D\|_p^p \right), \label{eq:expansion} \end{align} so that approximately minimizing \cref{eq:expansion} over circulations $\D$ suffices to make multiplicative progress towards the optimum of $\min_{B^Tf=d} \|f\|_p^p.$ We show in \cref{sec:refine} that a general class of objectives may all be reduced to approximately minimizing objectives which are combinations of quadratic and $\ell_p$ norm pieces. Specifically, any convex function $h$ with stable second derivative admits an expansion for $h(x+\D)^p - h(x)^p$ similar to \cref{eq:expansion}, which we show in \cref{lemma:iterh}. We then combine this expansion with the almost linear time smoothed $\ell_2$-$\ell_p$ solver of \cite{KPSW19} to show \cref{thm:mainopt}.
\section{Efficient Divergence Maximization}
\label{sec:energymax}
In this section we present our divergence maximization framework and show \cref{thm:main}.
Throughout the section we set $\eta = \log_m\left(m^{1/6-o(1)}U^{1/3}\right)$, in pursuit of a $m^{3/2-\eta+o(1)}$ time algorithm. We assume $G$ is preconditioned \cref{sec:precondition} and we maintain the invariant $\|w\|_1 \le 5m/2$ before each step, and $\|w\|_1 \le 3m$ at all times. As our algorithm runs in time at least $m^{3/2}$ for $U > \sqrt{m}$ we assume $U \le \sqrt{m}$ throughout this section. We assume that our algorithm knows $F_t$, as discussed in \cref{sec:setup}, and that $F_t \ge m^{1/2-\eta}$ or else we can round to an integral flow and run augmenting paths.

\paragraph{Notational conventions.} We largely adopt the same notation as used in \cref{sec:quadsmooth}. We use $D$ to refer to functions which are a Bregman divergence, and for a function $h$, we use $\widetilde{h}$ to denote a quadratic extension of $h$. For flows we use $f$ and $\hf$, the latter which refers to flows we wish to augment by. We use $\D$ and $\tD$ for circulations. The letters $w, \mu, \nu$ refer to weights and weight changes, and $W$ refers to a weight budget. \\

For $\eps < 1$ define the functions
\begin{align}
\label{eq:defphi}
\phi(x) &\defeq D_{-\log(1-x)}(x\| 0) = -\log(1-x)-x \enspace \text{ and } \enspace 
	\\ \tphi_\eps(x) &\defeq \begin{cases}
    \phi(x) & \text{ for } |x| \le \eps\\
    \phi(\eps)+\phi'(\eps)(x-\eps)+\frac{\phi''(\eps)}{2}(x-\eps)^2 & \text{ for } x \ge \eps \\
    \phi(-\eps)+\phi'(-\eps)(x+\eps)+\frac{\phi''(-\eps)}{2}(x+\eps)^2 & \text{ for } x \le -\eps.
  \end{cases}
\label{eq:deftphi}
\end{align}
Throughout, we will omit $\eps$ and write $\tphi(x) \defeq \tphi_{1/10}(x)$. As discussed in \cref{sec:quadsmooth} $\tphi(x)$ is such that it behaves as a quadratic around $x = 0$, and has continuous second derivatives. We have defined $\phi$ as the Bregman divergence of $-\log(1-x)$ from 0, and $\tphi$ is the quadratic extension of $\phi$ as described in \cref{sec:quadsmooth}. Several useful properties of the derivatives and stability of $\phi$ and $\tphi$ are collected in \cref{lemma:propphi}, which we prove in \cref{sec:refine}.
\begin{lemma}[Properties of $\tphi$]
\label{lemma:propphi}
We have that $1/2 \le \tphi''(x) \le 2$. Also, for $x \ge 0$ we have that $x/2 \le \tphi'(x) \le 2x$ and $-x/2 \ge \tphi'(-x) \ge -2x.$ We have that $x^2/4 \le \tphi(x) \le x^2$ for all $x$.
\end{lemma}
Now, we define the higher order analog of electric energy which we maximize under a weight budget. Below, we assume without loss of generality that $\cpe \le \cme$ for all edges $e$, as the orientation of each edge is arbitrary in the algorithm. In this way, $c_e = \cpe$ for all $e$.
\begin{equation}
\E_w(f) \defeq \sum_{e\in E} \left(\wpe\phi\left(\frac{f_e}{\cpe}\right)+\wme\phi\left(-\frac{f_e}{\cme}\right)\right) \label{eq:energy} \nonumber
\end{equation}
\begin{equation}
\tE(f) \defeq \sum_{e\in E} \left(\wpe\tphi\left(\frac{f_e}{\cpe}\right)+\wme\tphi\left(-\frac{f_e}{\cme}\right)\right) \label{eq:tenergy} \nonumber
\end{equation}

Define $\val$ and $\tval$ as follows, where $p = 2\lceil\sqrt{\log m}\rceil$, and $W = m^{6\eta}$ is a constant. For clarity, we express the vector inside the $\|\cdot\|_p$ piece of \cref{eq:val,eq:tval} coordinate-wise, where the coordinate corresponding to edge $e$ is written.
\begin{equation}
\val(f) \defeq \E_w(f) + W\left\|\left(\cpe\right)^2\left(\phi\left(\frac{f_e}{\cpe}\right) + \left(\frac{\cme}{\cpe}\right)\phi\left(-\frac{f_e}{\cme}\right)\right)\right\|_p \label{eq:val}
\end{equation}
\begin{equation}
\tval(f) \defeq \tE(f) + W\left\|\left(\cpe\right)^2\left(\tphi\left(\frac{f_e}{\cpe}\right) + \left(\frac{\cme}{\cpe}\right)\tphi\left(-\frac{f_e}{\cme}\right)\right)\right\|_p \label{eq:tval}
\end{equation}
These are defined so that for $q$ as the dual norm of $p$, i.e. $1/q + 1/p = 1$, we have that $\val$ and $\tval$ correspond to maximizing the minimum values of $\E_w(f)$ and $\tE(f)$ under a weighted $\ell_q$ weight budget. Specifically, we can compute using Sion's minimax theorem that
\begin{align} \label{eq:sion} &\max_{\substack{\|(\cpe)^{-2}\npe\|_q \le W \\ \nu \in \R_{\ge0}^{2E} \\ \frac{\npe}{\cpe} = \frac{\nme}{\cme} \forall e \in E}} \min_{B^Tf=d} \E_{w+\nu}(f) = \min_{B^Tf=d}\max_{\substack{\|(\cpe)^{-2}\npe\|_q \le W \\ \nu \in \R_{\ge0}^{2E} \\ \frac{\npe}{\cpe} = \frac{\nme}{\cme} \forall e \in E}} \E_{w+\nu}(f) \\
&= \min_{B^Tf=d}\max_{\substack{\|(\cpe)^{-2}\npe\|_q \le W \\ \nu \in \R_{\ge0}^{2E} \\ \frac{\npe}{\cpe} = \frac{\nme}{\cme} \forall e \in E}} \E_w(f) + \E_{\nu}(f) \\ &= \min_{B^Tf=d}\max_{\substack{\|(\cpe)^{-2}\npe\|_q \le W \\ \nu^+ \in \R_{\ge0}^{2E}}} \E_w(f) + \sum_{e\in E} \left(\npe\phi\left(\frac{f_e}{\cpe}\right)+\frac{\cme}{\cpe}\npe\phi\left(-\frac{f_e}{\cme}\right)\right) = \min_{B^Tf=d} \val(f) \end{align}
and similarly for $\tE(f)$ and $\tval(f)$. The objective requires the constraint $\npe / \cpe = \nme / \cme \forall e \in E$ to ensure that the weight increase $\nu$ maintains centrality, and the coefficient of $(\cpe)^{-2}$ in the weight budget $\|(\cpe)^{-2}\npe\|_q \le W$ is chosen to ensure that the $\ell_p$ piece in $\tval(f)$ is ensured to have approximately unit weights on the $f_e$. Precisely, by \cref{lemma:propphi} and $\cpe \le \cme$ we have \[ \left(\cpe\right)^2\left(\tphi\left(\frac{f_e}{\cpe}\right) + \left(\frac{\cme}{\cpe}\right)\tphi\left(-\frac{f_e}{\cme}\right)\right) = \Theta(f_e^2). \] We require this property in order to apply the smoothed $\ell_2$-$\ell_p$ flow solvers in \cref{thm:smoothflow}.

As $\min_{B^Tf=d} \val(f)$ is the result of applying Sion's minimax theorem to a saddle point problem, there will be an optimal solution pair $(f^*, \mu^*)$. Ultimately, $f^*$ will be the flow which we add to our current flow to arrive at the next central path point, and the weight change will be derived from applying a weight reduction procedure to $\mu^*$.

The remaining arguments in this section heavily use local optimality of convex functions. For this reason, we show that $\val(f)$ and $\tval(f)$ are convex in \cref{proofs:valconvex}.
\begin{lemma}
\label{lemma:valconvex}
$\val(f)$ and $\tval(f)$ are convex.
\end{lemma}

From now on, we fix a step size $\d = \frac{F_t}{10^5m^{1/2-\eta}}$.
This simplifies our analysis, as our objectives are no longer linear in $\d$, as is the case with electric flows. We now bound the 
minimum value of $\tE(f)$ over all $\d\chi$-flows, and show a congestion bound for the minimizer, where we recall that congestion of a flow is the ratio of flow on an edge to its residual capacity $c_e$. \cref{lemma:energybound,lemma:bounddual} generalize  corresponding bounds for electric flows shown in \cite{Madry16} and \cite{LS19} Lemma 4.5 and 5.2.
\begin{lemma}
\label{lemma:energybound}
Let $\d = \frac{F_t}{10^5m^{1/2-\eta}}$. Then $\min_{B^Tf=\d\chi} \tE(f) \le 5\cdot 10^{-7} m^{2\eta}.$
\end{lemma}
\begin{proof}
Let $f'$ be the flow which routes $\d / m$ units of flow on each of the $m$ preconditioning edges. For a preconditioning edge $e_p$, \cref{lemma:precon} and the invariant $\|w\|_1 \le 3m$  tells us that 
\[ \frac{f_{e_p}'}{c_{e_p}} \le \frac{\d/m}{F_t/7\|w\|_1} \le \frac{21}{10^5m^{1/2-\eta}} ~.
\]
Therefore, applying \cref{lemma:propphi} to $P$, the set of $m$ preconditioning edges, and again applying that $\|w\|_1 \le 3m$ gives us the desired bound, as
\begin{align*} \tE(f') &= \sum_{e \in P} \left(\wpe\tphi\left(\frac{f_e'}{\cpe}\right)+\wme\tphi\left(-\frac{f_e'}{\cme}\right)\right) \\ &\le \left(\frac{f_{e_p}'}{c_{e_p}}\right)^2 \sum_{e\in P}\left(\wpe+\wme\right) \le \left(\frac{21}{10^5m^{1/2-\eta}}\right)^2\|w\|_1 \le 5 \cdot 10^{-7} m^{2\eta} \end{align*}  
\end{proof}
\begin{lemma}
\label{lemma:bounddual}
Let $\d = \frac{F_t}{10^5m^{1/2-\eta}}$, and $\hf = \argmin_{B^Tf=\d\chi} \tE(f)$. Then $|\hf_e|c_e^{-2} \le m^{2\eta}.$
\end{lemma}
\begin{proof}
Local optimality tells us that there is a vector $z \in \R^V$ satisfying $Bz = \g\tE(\hf).$ This, \cref{lemma:energybound}, and \cref{lemma:propphi}, specifically that $x\tphi'(x) \le 2x^2 \le 8\tphi(x)$ gives us
\[ \d\chi^Tz = \hf^TBz = \hf^T\g\tE(\hf) = \sum_{e\in E} \hf_e\left(\frac{\wpe}{\cpe}\tphi'\left(\frac{\hf_e}{\cpe}\right)-\frac{\wme}{\cme}\tphi'\left(-\frac{\hf_e}{\cme}\right)\right) \le 8\tE(\hf). \]
Note that the flow $\hf$ is acyclic, i.e. there is no cycle where the flow is in the positive direction for every cycle edge, as decreasing the flow along a cycle reduces the objective value. Also, for all edges $e=(u,v)$, we have $z_v-z_u = [Bz]_e = [\g \tE(\hf)]_e$, which has the same sign as $\hf_e$. As $\hf$ is acyclic, it can be decomposed into $a$-$b$ paths. Since, some path contains the edge $e$, we get that $|[Bz]_e| = |z_v-z_u| \le z_b-z_a = \chi^Tz.$ Using that $x\tphi'(x) \ge x^2/2$ from \cref{lemma:propphi} we get that
\[ |[Bz]_e| = |[\g\tE(\hf)]_e| \ge \frac{\wpe|\hf_e|}{2(\cpe)^2}+\frac{\wme|\hf_e|}{2(\cme)^2} \ge \frac12|\hf_e|c_e^{-2}. \]
Combining these gives us \[ |\hf_e|c_e^{-2} \le 2|[Bz]_e| \le 2\chi^Tz \le 16\d^{-1}\tE(\hf) \le m^{2\eta} \] after using \cref{lemma:energybound} and $\d = \frac{F_t}{10^5m^{1/2-\eta}} \ge 10^{-5}.$
\end{proof}
Now, we show that computing $\hf = \argmin_{B^Tf=\d\chi} \tval(f)$ gives us weight changes to control the congestion of the divergence maximizing flow. For clarity, the process is shown in \cref{algo:augment}.

\begin{algorithm}[h]
\caption{$\Augment(G,w,F_t,f_{t,w}^*)$. Takes a preconditioned undirected graph $G$ with maximum capacity $U$, weights $w \in \R^{2E}_{\ge1}$ with $\|w\|_1 \le 5m/2$, residual flow $F_t = t^*-t$, central path point $f_{t,w}^*$. Returns step size $\d$, weights $\nu$, and $\d\chi$-flow $\hf$ with $f_{t,w+\nu}^* = f_{t,w}^* + \hf.$}
$\eta \assign \log_m(m^{1/6-o(1)}U^{-1/3})$ \\
$\d \assign \frac{F_t}{10^5m^{1/2-\eta}}$ \Comment{Step size.} \\
$\cpe \assign \upe-[f_{t,w}^*]_e, \cme \assign \ume+[f_{t,w}^*]_e$. \Comment{Residual capacities.} \\
$W \assign m^{6\eta}$ \Comment{Weight budget.} \\
$\hf \assign \argmin_{B^Tf=\d\chi} \tval(f)$ \Comment{$\tval(f)$ implicitly depends on $W,\cpe,\cme$.} \label{line:hf} \\
$v \in \R^E$ defined as $v_e \assign \left(\cpe\right)^2\left(\tphi\left(\frac{\hf_e}{\cpe}\right) + \left(\frac{\cme}{\cpe}\right)\tphi\left(-\frac{\hf_e}{\cme}\right)\right)$ for all $e \in E$. \label{line:vvv} \\
$\mu \in \R^{2E}_{\ge0}$ defined as $\mpe \assign W(\cpe)^2 \cdot \frac{v_e^{p-1}}{\|v\|_p^{p-1}}$ and $\mme \assign \frac{\cme}{\cpe}\mpe$. \Comment{Preliminary weight change.} \label{line:mu} \\
Initialize $\nu \in \R^{2E}_{\ge0}$. \Comment{Reduced weight change, satisfies $\frac{\npe}{\cpe-\hf_e} - \frac{\nme}{\cme+\hf_e} = \frac{\mpe}{\cpe-\hf_e} - \frac{\mme}{\cme+\hf_e}$} \label{line:startnu} \\
\For{$e \in E$}{
	\If{$\frac{\mpe}{\cpe-\hf_e} - \frac{\mme}{\cme+\hf_e} \ge 0$}{
		$\npe \assign (\cpe-\hf_e)\left(\frac{\mpe}{\cpe-\hf_e} - \frac{\mme}{\cme+\hf_e}\right)$, $\nme \assign 0$ \label{line:midnu} \\
	}\Else{
		$\npe \assign 0$, $\nme \assign -(\cme+\hf_e)\left(\frac{\mpe}{\cpe-\hf_e} - \frac{\mme}{\cme+\hf_e}\right)$ \label{line:endnu} \\
	}
}
Return $(\d,\hf,\nu)$.
\label{algo:augment}
\end{algorithm}

Now, we analyze \cref{algo:augment}. We start by showing that $\hf = \argmin_{B^Tf=\d\chi} \tEv(f)$ for the weight change $\mu$ in line \ref{line:mu} of \cref{algo:augment}, and that $\mu$ has bounded $\ell_1$ norm. This essentially follows from duality in our setup in \cref{eq:sion}.
\begin{lemma}
\label{lemma:divmin}
Let parameters $\eta,W,\cpe,\cme,\d$, flow $\hf$, and weight change $\mu$ be defined as in \cref{algo:augment}, and assume that $F_t \ge m^{1/2-\eta}.$ Then we have that $\|\mu\|_1 \le m/2$, $f_{t,w}^* = f_{t,w+\mu}^*$, and $\hf = \argmin_{B^Tf=\d\chi} \tEv(f).$
\end{lemma}
\begin{proof}
Let $v \in \R^E$ be the vector as defined in line \ref{line:vvv} in \cref{algo:augment}.
By local optimality of $\hf$, we have that there is a vector $z$ satisfying for all $e \in E$ that 
\begin{align} [Bz]_e &= \left[\g\tval(\hf)\right]_e \nonumber \\
&= \left(\frac{\wpe}{\cpe} + \frac{v_e^{p-1}}{\|v\|_p^{p-1}} \cdot \frac{W(\cpe)^2}{\cpe}\right)\tphi'\left(\frac{\hf_e}{\cpe}\right) - \left(\frac{\wme}{\cme} + \frac{v_e^{p-1}}{\|v\|_p^{p-1}} \cdot \frac{W\cpe\cme}{\cme}\right)\tphi'\left(-\frac{\hf_e}{\cme}\right). \label{eq:weightchange}
\end{align}
For clarity, we rewrite line \ref{line:mu} of \cref{algo:augment} here as
\begin{align}
\mpe = W(\cpe)^2 \cdot \frac{v_e^{p-1}}{\|v\|_p^{p-1}} \text{ and } \mme = \frac{\cme}{\cpe}\mpe = W\cpe\cme \cdot \frac{v_e^{p-1}}{\|v\|_p^{p-1}}. \label{eq:stillcentral}
\end{align}
Note that $\mpe/\cpe = \mme/\cme$, hence $f_{t,w}^* = f_{t,w+\mu}^*$ by \cref{eq:centralpath}. Combining \cref{eq:weightchange,eq:stillcentral} and optimality conditions of the objective $\min_{B^Tf=\d\chi} \widetilde{D^V_{w+\mu}}(f)$ shows that $\hf = \argmin_{B^Tf=\d\chi} \widetilde{D^V_{w+\mu}}(f)$. Additionally, if $q$ is the dual of $p$, i.e. $1/q+1/p=1$, then
\[ \|\mu\|_1 \le m^{o(1)}\|\mu\|_q \le 2m^{o(1)}WU^2\left\|\frac{v_e^{p-1}}{\|v\|_p^{p-1}}\right\|_q = 2m^{o(1)}WU^2 = m^{6\eta+o(1)}U^2 \le m/2 \] by our choice of $\eta = \log_m\left(m^{1/6-o(1)}U^{1/3}\right)$.
\end{proof}
We now show congestion bounds on $\hf$ by imitating the proof of \cref{lemma:energybound} and applying \cref{lemma:bounddual}. Recall that $c_e = \min(\cpe,\cme).$
\begin{lemma}
\label{lemma:congbound}
Let parameters $\eta,W,\cpe,\cme,\d$, flow $\hf$, and weight change $\mu$ be defined as in \cref{algo:augment}, and assume that $F_t \ge m^{1/2-\eta}.$ Then we have $|\hf_e| \le \frac{1}{500}m^{-2\eta}$ and $|\hf_e| \le \frac{1}{20}c_e$ for all edges $e$. It follows that $\hf = \argmin_{B^Tf=\d\chi} \val(f).$
\end{lemma}
\begin{proof}
We first show $\tval(\hf) \le 10^{-6}m^{2\eta}.$ Let $f'$ be the flow which routes $\frac{\d}{m}$ units of flow on each of the $m$ preconditioning edges. As in \cref{lemma:energybound} we have that $\tE(f') \le 5\cdot 10^{-7} m^{2\eta}.$ For a preconditioning edge $e$, using \cref{lemma:propphi} and \cref{lemma:precon} gives that
\begin{align*} &\left(\cpe\right)^2\left(\tphi\left(\frac{f_e'}{\cpe}\right) + \left(\frac{\cme}{\cpe}\right)\tphi\left(-\frac{f_e'}{\cme}\right)\right) \le \left(\cpe\right)^2\left(\left(\frac{f_e'}{\cpe}\right)^2 + \left(\frac{\cme}{\cpe}\right)\left(\frac{f_e'}{\cme}\right)^2\right) \le 2(f_e')^2 \\
&\le 2(\d/m)^2 \le 2m^{-2}\left(\frac{F_t}{10^5m^{1/2-\eta}}\right)^2 \le 10^{-8}m^{2\eta-1}U^2.
\end{align*}
as $F_t \le 3mU$. For the choice $W = m^{6\eta}$ we get that
\begin{align*} \tval(f') &\le \tE(f')+W\left\|\left(\cpe\right)^2\left(\tphi\left(\frac{f_e'}{\cpe}\right) + \left(\frac{\cme}{\cpe}\right)\tphi\left(-\frac{f_e'}{\cme}\right)\right)\right\|_p \\ &\le 5\cdot 10^{-7} m^{2\eta} + 10^{-8}m^{8\eta-1+o(1)}U^2 \le 10^{-6}m^{2\eta} \end{align*}
where we have used $\|x\|_p \le m^{o(1)}\|x\|_\infty$ for the choice $p = 2\left\lceil\sqrt{\log m}\right\rceil$, and the choice of $\eta = \log_m\left(m^{1/6-o(1)}U^{1/3}\right)$ to get $m^{8\eta-1+o(1)}U^2 \le m^{2\eta}$.

We now show $|\hf_e| \le \frac{1}{500}m^{-2\eta}$ for all $e$. Indeed, applying $\tphi(x) \ge x^2/4$ from \cref{lemma:propphi} yields
\[ \frac14W\hf_e^2 \le \tval(f) \le 10^{-6}m^{2\eta}. \] Using the choice $W = m^{6\eta}$ and rearranging gives us $|\hf_e| \le \frac{1}{500}m^{-2\eta}.$

Let $\mu$ be the weight increases given by line \ref{line:mu} of \cref{algo:augment}, and \cref{eq:stillcentral}. As $\hf = \min_{B^Tf=\d\chi} \widetilde{D^V_{w+\mu}}(f)$ and $\|w+\mu\|_1 \le 5m/2 + m/2 \le 3m$ by our invariant and \cref{lemma:divmin}, using \cref{lemma:bounddual} gives us
\[ |\hf_e|c_e^{-1} = \left(|\hf_e| \cdot |\hf_e|c_e^{-2}\right)^{1/2} \le \left(\frac{1}{500}m^{-2\eta} \cdot m^{2\eta}\right)^{1/2} \le \frac{1}{20}. \]
Using that the functions $\phi(x)$ and $\tphi(x)$ agree for $|x| \le \frac{1}{10}$, and $|\hf_e| \le \frac{1}{20}c_e$ for all $e$, \cref{obs:convex} gives us that $\hf$ is also a minimizer of $\min_{B^Tf=\d\chi} \val(f)$ as desired.
\end{proof}
We now show that applying weight change $\mu$ and adding $\hf$ to our current central path point $f_{t,w}^*$ stays on the central path, for path parameter $t+\d$. This follows from optimality conditions on the objective, which we designed to satisfy exactly the desired property.

As the weight change $\mu$ may be too large, we reduce the weight change $\mu$ to a weight change $\nu$ after advancing the path parameter, and bound $\|\nu\|_1.$ Intuitively, this weight reduction procedure can never hurt the algorithm. It happens to help because we carefully designed our objective to induce smoothed $\ell_2$-$\ell_p$ flow instances with unit weights on the $\ell_p$ part, the only instances which are known to admit almost linear runtimes \cite{KPSW19}, which we use in \cref{sec:efficient}.\footnote{Even with an oracle for smoothed $\ell_2$-$\ell_p$ flow that handles arbitrary weights, we do not know how to achieve maxflow runtimes faster than those achieved by this paper.}

\begin{lemma}
\label{lemma:finalweight}
Let parameters $\eta,W,\cpe,\cme,\d$, flow $\hf$, and weight changes $\mu,\nu$ be defined as in \cref{algo:augment}, and assume that $F_t \ge m^{1/2-\eta}.$ Then we have that $\|\nu\|_1 \le m^{4\eta+o(1)}U$ and $f_{t+\d,w+\nu}^* = f_{t,w}^* + \hf.$
\end{lemma}
\begin{proof}
We first show $f_{t+\d,w+\mu}^* = f_{t,w+\mu}^* + \hf = f_{t,w}^* + \hf.$ By \cref{lemma:congbound}, we have $\hf = \argmin_{B^Tf=\d\chi} \val(f)$. Let the vector $v$ be defined as in line \ref{line:vvv} of \cref{algo:augment}. There exist vectors $y, z \in \R^E$ such that
\begin{align*} [Bz]_e &= \left[\g\val(\hf)\right]_e \\
&= \left(\frac{\wpe}{\cpe} + \frac{v_e^{p-1}}{\|v\|_p^{p-1}} \cdot \frac{W(\cpe)^2}{\cpe}\right)\phi'\left(\frac{\hf_e}{\cpe}\right) - \left(\frac{\wme}{\cme} + \frac{v_e^{p-1}}{\|v\|_p^{p-1}} \cdot \frac{W\cpe\cme}{\cme}\right)\phi'\left(-\frac{\hf_e}{\cme}\right) \\
&= \left[\frac{\wpe+\mpe}{\cpe-\hf_e} - \frac{\wpe+\mpe}{\cpe}\right] - \left[\frac{\wme+\mme}{\cme+\hf_e} - \frac{\wme+\mme}{\cme}\right] \\
&= \left[\frac{\wpe+\mpe}{\upe-\left[f_{t,w+\mu}^*\right]_e-\hf_e} - \frac{\wme+\mme}{\ume+\left[f_{t,w+\mu}^*\right]_e+\hf_e}\right] - [By]_e.
\end{align*}
Here, the first line follows from local optimality of $\hf = \argmin_{B^Tf=\d\chi} \val(f)$, the third is explicit computation of $D'$, and the fourth follows from centrality of $f_{t,w+\mu}^*$. Therefore, the $(t+\d)\chi$-flow which is $f_{t,w+\mu}^*+\hf$ satisfies
\[ \left[\frac{\wpe+\mpe}{\upe-\left[f_{t,w+\mu}^*\right]_e-\hf_e} - \frac{\wme+\mme}{\ume+\left[f_{t,w+\mu}^*\right]_e+\hf_e}\right] = \left[B(y+z)\right]_e \] hence is central for weights $w+\mu$. So $f_{t+\d,w+\mu}^* = f_{t,w+\mu}^* + \hf = f_{t,w}^* + \hf.$

Now, note that $\nu$ as defined in lines \ref{line:startnu} to \ref{line:endnu} of \cref{algo:augment} satisfies
\[ \frac{\npe}{\cpe-\hf_e} - \frac{\nme}{\cme+\hf_e} = \frac{\mpe}{\cpe-\hf_e} - \frac{\mme}{\cme+\hf_e}
 \] and centrality conditions \cref{eq:centralpath} tell us that $f_{t+\d,w+\nu}^* = f_{t+\d,w+\mu}^*$. 
 
We now bound $\|\nu\|_1.$ Line \ref{line:endnu} of \cref{algo:augment} and $\mpe/\cpe = \mme/\cme$ gives us that
\[ 
\npe+\nme = -(\cme + \hf_e)\left(\frac{\mpe}{\cpe-\hf_e} - \frac{\mme}{\cme+\hf_e}\right) = -\mme\left(\left(\frac{\cme+\hf_e}{\cpe-\hf_e}\right)\frac{\cpe}{\cme}-1\right) 
\le 3c_e^{-1}|\hf_e|\mme, 
\] 
where we have used that $c_e^{-1}|\hf_e| \le 1/20$. A similar analysis of line \ref{line:midnu} gives that
\[ \npe+\nme = (\cpe - \hf_e)\left(\frac{\mpe}{\cpe-\hf_e} - \frac{\mme}{\cme+\hf_e}\right) = \mpe\left(1 - \left(\frac{\cpe-\hf_e}{\cme+\hf_e}\right) \frac{\cme}{\cpe}\right) \le 3c_e^{-1}|\hf_e|\mpe. \] In both cases, we have that $\npe+\nme \le 3c_e^{-1}|\hf_e|(\mpe+\mme).$
Using the choice $W=m^{6\eta}$, \cref{eq:stillcentral,lemma:congbound} yield
\begin{align*}
\|\nu\|_1 &\le \sum_{e \in E} 3c_e^{-1}|\hf_e|(\mme + \mpe) \le 6W\sum_{e\in E} |\hf_e|\cme \cdot \frac{v_e^{p-1}}{\|v\|_p^{p-1}} \\
&\le 12W \cdot \frac{1}{500}m^{-2\eta}U \left\|\frac{v_e^{p-1}}{\|v\|_p^{p-1}}\right\|_1 \le m^{4\eta+o(1)}U\left\|\frac{v_e^{p-1}}{\|v\|_p^{p-1}}\right\|_q = m^{4\eta+o(1)}U.
\end{align*}
\end{proof}
\subsection{Efficient Divergence Maximization}
\label{sec:efficient}
\cref{lemma:finalweight} shows that our algorithm just needs to compute $\argmin_{B^Tf=\d\chi} \tval(f)$ in line \ref{line:hf} of \cref{algo:augment}, as all other lines clearly take $O(m)$ time. Here, we show how to do this in time $m^{1+o(1)}$.
\begin{lemma}
\label{lemma:fasttval}
There is an algorithm that in $m^{1+o(1)}$ time computes a flow $f'$ with $B^Tf'=\d\chi$ and $\tval(f') \le \min_{B^Tf=\d\chi} \tval(f) + \err$.
\end{lemma}
To prove \cref{lemma:fasttval}, we first show \cref{thm:mainopt} by carefully applying and extending the analysis of the following result of \cite{KPSW19} on smoothed $\ell_2$-$\ell_p$ flows.
\begin{theorem}[Theorem 1.1 in \cite{KPSW19}, arXiv version]
\label{thm:smoothflow}
Consider $p \in (\omega(1), (\log n)^{2/3-o(1)})$, $g \in \R^E$, $r \in \R^E_{\ge0}$, demand vector $d\in \R^V$, real number $s \ge 0$, and initial solution $f_0\in\R^E$ such that all parameters are bounded by $2^{\poly(\log m)}$ and $B^Tf_0=d$. For a flow $f$, define 
\[ \vallp(f) \defeq \sum_{e\in E} g_ef_e + \left(\sum_{e \in E} r_ef_e^2\right) + s\|f\|_p^p
\enspace \text{ and } \enspace
OPT \defeq \min_{B^Tf=d} \vallp(f). \]
There is an algorithm that in $m^{1+o(1)}$ time computes a flow $f$ such that $B^Tf=d$ and
\[ \vallp(f)-OPT \le \err(\vallp(f_0)-OPT) + \err. \]
\end{theorem}
While Theorem 1.1 in \cite{KPSW19} is stated with $\frac{1}{\poly(m)}$ errors, it can be made to $\err$ as explained in \cite{LS19} Appendix D.3.
While we defer the full proof of \cref{thm:mainopt} to \cref{sec:opt}, we give a brief proof sketch here. At a high level, we first use Lemma B.3 of \cite{LS19} to reduce optimizing the objective $\val(f)$ in \cref{thm:mainopt} to solving $\O(1)$ problems of the following kind, through careful binary search on $W$. For
\[ \val_{p,W}(f) \defeq \sum_{e\in E} q_e(f_e) + W\sum_{e\in E} h_e(f_e)^p \enspace \text{ and } \enspace OPT_{p,W} \defeq \min_{B^Tf=d} \val_{p,W}(f) \] find a flow $f'$ with $B^Tf'=d$ and $\val_{p,W}(f') \le OPT + \err.$ This is done formally in \cref{lemma:reduce1}.

We can then apply the iterative refinement framework on the objective $\val_{p,W}(f)$ to reduce it to solving $\O(2^{O(p)})$ smoothed quadratic and $\ell_p$ norm flow problems, which may be solved using \cref{thm:smoothflow}. The main difference from the analysis of \cite{KPSW19} is that we show that any convex function $h$ with stable second derivatives admits an expansion for $h(x+\D)^p - h(x)^p$, while \cite{KPSW19} only considers the function $h(x) = x^2$. This is done formally in \cref{lemma:reduce2}.

\begin{proof}[Proof of \cref{lemma:fasttval}]
It suffices to show that $\tval(f)$ satisfies the constraints of \cref{thm:mainopt} for 
\begin{align*}
q_e(x) &\defeq \wpe\tphi\left(\frac{x}{\cpe}\right)+\wme\tphi\left(-\frac{x}{\cme}\right) \enspace \text{ and } \enspace a_e \defeq \left(\frac{\wpe}{(\cpe)^2}+\frac{\wme}{(\cme)^2}\right),
\\ h_e(x) &\defeq \left(\cpe\right)^2\left(\tphi\left(\frac{x}{\cpe}\right) + \left(\frac{\cme}{\cpe}\right)\tphi\left(-\frac{x}{\cme}\right)\right).
\end{align*}
To analyze $q_e(x)$, we compute that \[ q_e''(x) = \frac{\wpe}{(\cpe)^2}\tphi''\left(\frac{x}{\cpe}\right) + \frac{\wme}{(\cme)^2}\tphi''\left(\frac{x}{\cme}\right) \le 2\left(\frac{\wpe}{(\cpe)^2}+\frac{\wme}{(\cme)^2}\right) = 2a_e \] by \cref{lemma:propphi}. A lower bound $q_e''(x) \ge \frac12\left(\frac{\wpe}{(\cpe)^2}+\frac{\wme}{(\cme)^2}\right) = a_e/2$ follows equivalently. $a_e \le 2^{\poly(\log m)}$ follows from the fact that resistances and residual capacities are polynomially bounded on the central path -- see \cite{LS19} Lemma D.1.

To analyze $h_e(x)$, we can compute using \cref{eq:deftphi} that $h_e(0) = h_e'(0) = 0$ and \[ h_e''(x) = \tphi''\left(\frac{x}{\cpe}\right) + \frac{\cpe}{\cme}\tphi''\left(-\frac{x}{\cme}\right) \le 4 \] by \cref{lemma:propphi}. We get $h_e''(x) \ge 1/2$ similarly.
\end{proof}

\subsection{Algorithm}
Here we state \cref{algo:maxflow} to show \cref{thm:main}. It repeatedly takes steps computed with \cref{algo:augment}, and when the remaining flow is $m^{1/2-\eta}$, we stop the algorithm, round to an integral flow and run augmenting paths.

\begin{algorithm}[h]
\caption{$\Maxflow(G)$. Takes a preconditioned undirected graph $G$ with maximum capacity $U$. Returns the maximum $ab$ flow in $G$.}
$\eta \assign \log_m(m^{1/6-o(1)}U^{-1/3}).$ \\
$f \assign 0, t \assign 0, w \assign 1$. \\
\While{$F_t \ge m^{1/2-\eta}$}{ \label{line:mainwhile}
	$(\d,\hf,\nu) \assign \Augment(G,w,t^*-t,f)$. \label{line:2} \\
	$f \assign f+\hf, w \assign w+\nu$, and $t \assign t+\d$. \\
}
Round to an integral flow and use augmenting paths until done.
\label{algo:maxflow}
\end{algorithm}

\begin{proof}[Proof of \cref{thm:main}]
We show that $\Maxflow(G)$ computes a maximum flow on $G$ in time \\ $m^{3/2-\eta+o(1)}$ $= m^{4/3+o(1)}U^{1/3}$ by the choice of $\eta$. Correctness follows from \cref{lemma:finalweight,lemma:fasttval}. 

It suffices to control the weights. Our choice of $\d$ guarantees that we route $\Omega(m^{-1/2+\eta})$ fraction of the remaining flow per iteration, hence line \ref{line:mainwhile} executes $\O(m^{1/2-\eta})$ times. $\|\nu\|_1 \le m^{4\eta+o(1)}U$ always by \cref{lemma:finalweight}, hence at the end of the algorithm by the choice of $\eta$ we have
\[ 
\|w\|_1 \le 2m + \O\left(m^{1/2-\eta} \cdot m^{4\eta+o(1)}U \right) \le 5m/2 
~.
\] 

To analyze the runtime, first note that by \cref{lemma:fasttval} line \ref{line:2} takes $m^{1+o(1)}$ time, so the total runtime of these throughout the algorithm is $m^{3/2-\eta+o(1)}$. Rounding to an integral flow takes $\O(m)$ time \cite{LRS13,Madry13}. Augmenting paths takes $O(m^{3/2-\eta})$ time also, as desired.
\end{proof}
\section{Conclusion}
\label{sec:conclusion}

We conclude by first stating the difficulties in going beyond an $m^{4/3}$ runtime for maxflow, and then discussing potential directions for future research on the topic.

\paragraph{A Possible Barrier at $m^{4/3}$:}
Here we briefly discuss why we believe that $m^{4/3}$ is a natural runtime barrier for IPM based algorithms for maxflow on sparse unweighted graphs. All currently known weighted IPM advances satisfy the following two properties. First, weights only increase and do not be become super linear. Second, the methods step from one central path point to the next one in $\Theta(m)$ time and the congestion of this step is multiplicatively bounded.
Our algorithm does both these pieces optimally --- we precisely compute weight changes under a budget to ensure that the congestion to the next central path point is reduced significantly. In this sense, to break the $m^{4/3}$ barrier, one would either have to find a new way to backtrack on weight changes on the central path so that they are not additive throughout the algorithm, or show better amortized bounds on weight change than shown here. Alternatively, one would need a novel algorithm to step to faraway points on the central path, outside a ball where the congestions are bounded.

\paragraph{Potential future directions:}
We believe that there are several exciting open problems remaining, including combining this work with previous results of \cite{CMSV16,AMV20} to achieve faster algorithms for mincost flow, as well as potentially achieving an $mn^{1/3+o(1)}$ time algorithm for maxflow through the approach of \cite{LeeS19arXiv} and robust central paths \cite{CohenLS19,BLSZ20}.

Additionally, it would be interesting to understand whether IPMs can achieve $m^{3/2-\Omega(1)}\log^{O(1)}U$ runtimes for maxflow. Also, is there an algorithm for $m^{4/3+o(1)}$ time exact directed maxflow in unit capacity graphs that only uses Laplacian system solvers? Finally, on the other hand, can we use stronger primitives such as smoothed $\ell_2$-$\ell_p$ flows to design even more efficient algorithms for unit capacity maxflow, potentially beyond the use of IPMs?

\section{Acknowledgements}
\label{sec:acknowledgements}
We thank Arun Jambulapati, Tarun Kathuria, Michael B. Cohen, Yin Tat Lee, Jonathan Kelner, Aleksander M\k{a}dry, and Richard Peng for helpful discussions. The idea of quadratic extensions was originally suggested by Yin Tat Lee, and we are extremely grateful to him for this fruitful insight.

{\small
\bibliographystyle{alpha}
\bibliography{refs}}

\appendix

\section{Missing proofs}
\label{sec:proofs}
\subsection{Proof of \cref{obs:convex}}
\label{proofs:obsconvex}
\begin{proof}
Assume for contradiction that $g(x) < g(x^*)$ for some $x \in \chi$. For all $\eps > 0$ we have \[ g((1-\eps)x^* + \eps x) \le (1-\eps)g(x^*) + \eps g(x) < g(x^*). \] Because $f, g$ agree on a neighborhood of $x^*$, we have for sufficiently small $\eps$ that
\[ f((1-\eps)x^* + \eps x) = g((1-\eps)x^* + \eps x) < g(x^*) = f(x^*), \] a contradiction to $f(x^*) = \min_{x \in \chi} f(x)$, as $(1-\eps)x^* + \eps x \in \chi$ by convexity of $\chi$.
\end{proof}

\subsection{Proof of \cref{lemma:valconvex}}
\label{proofs:valconvex}
\begin{lemma}
\label{lemma:convexlemma}
Let $p \ge 1$ be a real number. Let $h_i:\R \to \R_{\ge0}$ be convex functions. Then the function $h:\R^n\to\R$ defined by $h(x) = \left(\sum_{i=1}^n h_i(x_i)^p\right)^{1/p}$ is a convex function.
\end{lemma}
\begin{proof}
For any $x,y\in\R^n$ and $0 \le t \le 1$ we have by Minkowski's inequality and convexity that
\begin{align*}h(tx+(1-t)y) &= \|h_i(tx_i+(1-t)y_i)\|_p \le \|t \cdot h_i(x_i) + (1-t) \cdot h_i(y_i)\|_p \\ &\le t\|h_i(x_i)\|_p+(1-t)\|h_i(y_i)\|_p = t \cdot h(x) + (1-t) \cdot h(y). \end{align*}
\end{proof}
\begin{proof}[Proof of \cref{lemma:valconvex}]
$\phi$ and $\tphi$ are convex, hence $\E_w(f)$ and $\tE(f)$ are convex. Also, the functions \[ \left(\spe\right)^2\left(\phi\left(\frac{f_e}{\spe}\right) + \left(\frac{\sme}{\spe}\right)\phi\left(-\frac{f_e}{\sme}\right)\right) \enspace \text{ and } \enspace \left(\spe\right)^2\left(\tphi\left(\frac{f_e}{\spe}\right) + \left(\frac{\sme}{\spe}\right)\tphi\left(-\frac{f_e}{\sme}\right)\right) \] are convex, hence $\val(f)$ and $\tval(f)$ are convex by \cref{lemma:convexlemma}.
\end{proof}

\section{Iterative Refinement}
\label{sec:refine}
Here, $\phi$ and $\tphi$ are defined as in \cref{eq:defphi}.
\begin{lemma}
\label{lemma:inth}
Let $h:\R\to\R$ be a function with $h(0) = h'(0) = 0$, and let $c_1, c_2 > 0$ be constants such that $c_1 \le h''(x) \le c_2$ for all $x$. Then for $x \ge 0$ we have that $c_1x \le h'(x) \le c_2x$ and $-c_2x \le h'(-x) \le -c_1x.$ Also, $\frac{1}{2}c_1x^2 \le h(x) \le \frac{1}{2}c_2x^2$ for all $x$.
\end{lemma}
\begin{proof}
For $x \ge 0$ we have that $h'(x) = \int_0^x h''(y) dy$ and $c_1x \le \int_0^x h''(y) dy \le c_2x$. The proof for $x \le 0$ is equivalent.

For $x \ge 0$ we have that $h(x) = \int_0^x h'(y) dy$ and \[ \frac12c_1x^2 = \int_0^x c_1y dy \le \int_0^x h'(y) dy \le \int_0^x c_2y dy = \frac12c_2x^2. \] The proof for $x \le 0$ is equivalent.
\end{proof}
\begin{proof}[Proof of \cref{lemma:propphi}]
It suffices to show $1/2 \le \tphi''(x) \le 2$ and apply \cref{lemma:inth}.
For $|x| \le 1/10$ we have that $\tphi''(x) = (1-x)^{-2}$. Now, we can check that for $|x| \le 1/10$ that $1/2 \le (1-x)^{-2} \le 2$. For $x \ge 1/10$ we have $\tphi''(x) = \tphi''(1/10)$ and for $x \le -1/10$ we have $\tphi''(x) = \tphi''(-1/10)$ as desired.
\end{proof}

\begin{lemma}[Lemmas B.2 and B.3 from \cite{KPSW19}, arXiv version]
\label{lemma:kpswiter}
Let $p > 0$ be an even integer. For all real numbers $x, \D$ we have that
\[ 2^{-p}\left(x^{p-2}\D^2 + \D^{p}\right) \le (x+\D)^p - \left(x^p + p \cdot x^{p-1}\D\right) \le p2^{p-1}\left(x^{p-2}\D^2 + \D^{p}\right). \]
\end{lemma}

\begin{lemma}
\label{lemma:iterh}
Let $h:\R\to\R$ be a function and let $c_2 \ge c_1 > 0$ be constants such that $c_1 \le h''(x) \le c_2$ for all $x$. Then for all $\D$ we have that
\begin{equation} \frac12c_1\D^2 \le h(x+\D)-\left(h(x)+h'(x)\D\right) \le \frac12c_2\D^2. \label{eq:2iter} \end{equation}
If additionally $h(0) = h'(0) = 0$ and $c_1 \le h''(x) \le c_2$ for all $x$ then for all even integers $p > 0$ and $\D$ we have that
\begin{equation} (8c_2)^{-2p}c_1^{3p} \left(x^{2p-2}\D^2 + \D^{2p}\right) \le h(x+\D)^p - \left(h(x)^p + p \cdot h(x)^{p-1}h'(x)\D\right) \le (16c_2)^p \left(x^{2p-2}\D^2 + \D^{2p}\right) \label{eq:piter} \end{equation}
\end{lemma}
\begin{proof}
The upper bound of \cref{eq:2iter} follows from \[ h(x+\D)-\left(h(x)+h'(x)\D\right) = \int_0^\D (\D-y)h''(x+y) dy \le c_2 \int_0^\D (\D-y) dy = \frac12c_2\D^2 \] by \cref{lemma:inth}. The lower bound follows equivalently.

We show the upper bound of \cref{eq:piter}. Use \cref{lemma:kpswiter} with $x \to h(x), \D \to h(x+\D)-h(x)$ to get
\begin{align} &h(x+\D)^p - \left(h(x)^p + p \cdot h(x)^{p-1}\left(h(x+\D)-h(x)\right)\right) \nonumber \\ &\le p2^{p-1}\left(h(x)^{p-2}\left(h(x+\D)-h(x)\right)^2 + \left(h(x+\D)-h(x)\right)^p\right). \label{eq:bigterm} \end{align}
Using the upper bound of \cref{eq:2iter} and $h(x) \le \frac12c_2x^2$ from \cref{lemma:inth} gives us
\begin{align}
&h(x+\D)^p - \left(h(x)^p + p \cdot h(x)^{p-1}h'(x)\D\right) \label{eq:use1} \\ &\le h(x+\D)^p - \left(h(x)^p + p \cdot h(x)^{p-1}\left(h(x+\D)-h(x)\right)\right) + \frac12p \cdot h(x)^{p-1} c_2\D^2 \label{eq:use2} \\ &\le h(x+\D)^p - \left(h(x)^p + p \cdot h(x)^{p-1}\left(h(x+\D)-h(x)\right)\right) + p \cdot c_2^p x^{2p-2}\D^2. \label{eq:use3}
\end{align}
We now bound \cref{eq:bigterm}. \cref{eq:2iter} and \cref{lemma:inth} gives us \[ \left|h(x+\D)-h(x)\right| \le \left|h'(x)\D\right|+c_2\D^2 \le c_2|x\D|+c_2\D^2. \]
Using this and \cref{lemma:inth} gives us that \cref{eq:bigterm} is at most
\begin{align*}
&p2^{p-1}\left(h(x)^{p-2}\left(c_2|x\D|+c_2\D^2\right)^2 + \left(c_2|x\D|+c_2\D^2\right)^p\right) \\
&\le p2^{p-1}\left(x^{2p-4}c_2^{p-2}\left(2c_2^2x^2\D^2 + 2c_2^2\D^4\right) + 2^{p-1}c_2^px^p\D^p + 2^{p-1}c_2^p\D^{2p}\right) \\
&\le p2^{p-1}c_2^p \cdot \left(4+2^{p-1}+2^{p-1}\right) \cdot \left(x^{2p-2}\D^2 + \D^{2p}\right).
\end{align*}
Collecting terms and combining this with \cref{eq:use1,eq:use2,eq:use3} proves the upper bound of \cref{eq:piter}.

Now we show the lower bound of \cref{eq:piter}. As above, we use \cref{lemma:kpswiter} to get
\begin{align} &h(x+\D)^p - \left(h(x)^p + p \cdot h(x)^{p-1}\left(h(x+\D)-h(x)\right)\right) \nonumber \\ &\ge 2^{-p}\left(h(x)^{p-2}\left(h(x+\D)-h(x)\right)^2 + \left(h(x+\D)-h(x)\right)^p\right). \label{eq:bigterm2} \end{align}
Using the lower bound of \cref{eq:2iter} and $h(x) \ge \frac12c_1x^2$ from \cref{lemma:inth} gives us
\begin{align}
&h(x+\D)^p - \left(h(x)^p + p \cdot h(x)^{p-1}h'(x)\D\right) \label{eq:use4} \\ &\ge h(x+\D)^p - \left(h(x)^p + p \cdot h(x)^{p-1}\left(h(x+\D)-h(x)\right)\right) + \frac12p \cdot h(x)^{p-1}c_1\D^2 \label{eq:use5} \\ &\ge h(x+\D)^p - \left(h(x)^p + p \cdot h(x)^{p-1}\left(h(x+\D)-h(x)\right)\right) + p2^{-p}c_1^p \cdot x^{2p-2}\D^2. \label{eq:use6}
\end{align}
If $|x| \ge \frac{c_1|\D|}{4c_2}$ then \[ p2^{-p}c_1^p \cdot x^{2p-2}\D^2 \ge \frac12 \cdot p2^{-p}c_1^p\left(x^{2p-2}\D^2 + \left(\frac{c_1\D}{4c_2}\right)^{2p-2}\D^2\right) \ge (8c_2)^{-2p}c_1^{3p} \cdot \left(x^{2p-2}\D^2 + \D^{2p}\right) \] as desired.
If $|x| \le \frac{c_1|\D|}{4c_2}$ then applying the lower bound of \cref{eq:2iter} and \cref{lemma:inth} gives us
\[ h(x+\D)-h(x) \ge h'(x)\D + c_1\D^2/2 \ge -\left|c_2x\D\right|+c_1\D^2/2 \ge c_1\D^2/4. \]
Therefore, we may lower bound \cref{eq:bigterm2} by
\[ 2^{-p}\left(c_1\D^2/4\right)^p = 2^{-3p}c_1^p\D^{2p} \ge 2^{-3p-1}c_1^p\left(x^{2p-2}\D^2 + \D^{2p}\right) \]
for $|x| \le \frac{c_1|\D|}{4c_2} \le \D$. Combining this with \cref{eq:use4,eq:use5,eq:use6} gives the desired bound.
\end{proof}

\begin{lemma}
\label{lemma:refineprogress}
Let $p$ be an even positive integer. Let $h_e(x):\R\to\R$ be convex functions for $e \in E$, and for $x \in \R^E$ let $h(x) = \sum_{e\in E}h_e(x_e).$ 
Let $OPT = \min_{B^Tf=d} h(f).$ Let $f$ be a flow satisfying $B^Tf=d$. Let $C_1, C_2 > 0$ be constants such that for all edges $e$, there are real numbers $r_e, s_e \ge 0$ and $g_e$, depending on $f_e$, such that for all $\D_e \in \R$
\begin{equation} C_1\left(r_e\D_e^2+s_e\D_e^p\right) \le h_e(f_e+\D_e)-\left(h_e(f_e) + g_e\D_e\right) \le C_2\left(r_e\D_e^2+s_e\D_e^p\right). \label{eq:bounded} \end{equation}
Let \[ \tD = \argmin_{B^T\D=0} g^T\D + C_1\left(\sum_{e\in E}r_e\D_e^2+\sum_{e\in E}s_e\D_e^p\right). \]
Then \[ \left(h\left(f+\frac{C_1}{C_2}\tD\right) - OPT\right) \le \left(1 - \frac{C_1}{C_2}\right)\left(h(f)-OPT\right). \]
\end{lemma}
\begin{proof}
Define $f^* = \argmin_{B^Tf=d} h(f).$ Define $\D = f^*-f$. By the first inequality of \cref{eq:bounded}, we get
\begin{align*} &g^T\tD + C_1\left(\sum_{e\in E}r_e\tD_e^2+\sum_{e\in E}s_e\tD_e^p\right) \le g^T\D + C_1\left(\sum_{e\in E}r_e\D_e^2+\sum_{e\in E}s_e\D_e^p\right) \\ &\le h(f+\D)-h(f) = OPT-h(f). \end{align*}
This and the right side inequality of \cref{eq:bounded} give us
\begin{align*} &h\left(f+\frac{C_1}{C_2}\tD\right)-h(f) \le \frac{C_1}{C_2}g^T\tD + C_2\left(\left(\frac{C_1}{C_2}\right)^2\sum_{e\in E} r_e\tD_e^2+ \left(\frac{C_1}{C_2}\right)^p\sum_{e\in E} s_e\tD_e^p\right) \\ &\le \frac{C_1}{C_2}\left(g^T\tD + C_1\left(\sum_{e\in E}r_e\tD_e^2+\sum_{e\in E}s_e\tD_e^p\right)\right) \le \frac{C_1}{C_2}\left(OPT-h(f)\right).
\end{align*}
Rearranging this gives the desired inequality.
\end{proof}
\section{Additional Preliminaries}
\label{sec:optprelim}
In this section, we state some preliminaries for convex optimization. These will be used in \cref{sec:opt}. We assume all functions in this section to be convex. We also work in the $\ell_2$ norm exclusively. Proofs for the results stated can be found in \cite{Nes98}.

\paragraph{Matrices and norms.} We say that a $m\times m$ matrix $M$ is positive semidefinite if $x^TMx \ge 0$ for all $x \in \R^m.$ We say that $M$ is positive definite if $x^TMx > 0$ for all nonzero $x \in \R^m$. For $m \times m$ matrices $A, B$ we write $A \se B$ if $A-B$ is positive semidefinite, and $A \succ B$ is $A-B$ is positive definite. For $m \times m$ positive semidefinite matrix $M$ and vector $x \in \R^m$ we define $\|x\|_M = \sqrt{x^TMx}.$ For $m \times m$ positive semidefinite matrices $M_1, M_2$ and $C > 0$ we say that $M_1 \approx_C M_2$ if $\frac{1}{C} x^TM_1x \le x^TM_2x \le Cx^TM_1x$ for all $x \in \R^m.$

\paragraph{Lipschitz functions.} Here we define what it means for a function $f$ to be Lipschitz and provide a lemma showing its equivalence to a bound on the norm of the gradient.

\begin{definition}[Lipschitz Function]
Let $f: \R^n \to \R$ be a function, and let $\X \subseteq \R^n$ be an open convex set. We say that $f$ is $L_1$-Lipschitz on $\X$ (in the $\ell_2$ norm) if for all $x, y \in \X$ we have that $|f(x) - f(y)| \le L_1\|x-y\|_2.$
\end{definition}

\begin{lemma}[Gradient Characterization of Lipschitz Function]
\label{lemma:gradlip}
Let $f: \R^n \to \R$ be a differentiable function, and let $\X \subseteq \R^n$ be an open convex set. Then $f$ is $L_1$-Lipschitz on $\X$ if and only if for all $x \in \X$ we have that $\|\g f(x)\|_2 \le L_1.$
\end{lemma}

\paragraph{Smoothness and strong convexity.} We define what it means for a function $f$ to be convex, smooth, and strongly convex. We say that a function $f$ is convex on $\X$ if for all $x, y \in \X$ and $0 \le t \le 1$ that $f(tx+(1-t)y) \le tf(x)+(1-t)f(y).$ We say that $f$ is $L_2$-smooth on $\X$ if $\|\g f(x)-\g f(y)\|_2 \le L_2\|x-y\|_2$ for all $x, y \in \X$. We say that $f$ is $\mu$-strongly convex on $\X$ if for all $x, y \in \X$ and $0 \le t \le 1$ that
\[ f(tx+(1-t)y) \le tf(x)+(1-t)f(y)-t(1-t) \cdot \frac{\mu}{2}\|x-y\|_2^2. \]

\begin{lemma}
Let $f:\R^n \to \R$ be a differentiable function, and let $\X \subseteq \R^n$ be an open convex set. Then $f$ is $\mu$-strongly convex on $\X$ if and only if for all $x, y \in \X$ we have that \[ f(y) \ge f(x) + \g f(x)^T(y-x) + \frac{\mu}{2}\|y-x\|_2^2. \] Also, $f$ is $L_2$-smooth on $\X$ if and only if for all $x, y \in \X$ we have that \[ f(y) \le f(x) + \g f(x)^T(y-x) + \frac{L_2}{2}\|y-x\|_2^2. \]
\end{lemma}
We can equivalently view smoothness and strong convexity as spectral bounds on the Hessian of $f$.
\begin{lemma}
\label{lemma:equiv}
Let $f:\R^n \to \R$ be a twice differentiable function, and let $\X \subseteq \R^n$ be an open convex set. Then $f$ is $\mu$-strongly convex on a convex set $\X$ if and only if $\g^2 f(x) \se \mu I$ for all $x \in \X$. $f$ is $L_2$-smooth on $\X$ if and only if $\g^2 f(x) \pe L_2 I$ for all $x \in \X$.
\end{lemma}
Smoothness allows us to relate function error and the norm of the gradient.
\begin{lemma}
\label{lemma:smoothgrad}
Let $\X \subseteq \R^n$ be an open convex set, and let $f:\R^n \to \R$ be $L_2$-smooth on $\X.$ Define $x^* = \argmin_{x \in \R^n} f(x)$, and assume that $x^*$ exists and $x^* \in \X.$ Then for all $x \in X$ we have that
\[ \|\g f(x)\|_2^2 \le 2L_2(f(x)-f(x^*)). \]
\end{lemma}
Strong convexity allows us to relate function error and distance to the optimal point.
\begin{lemma}
\label{lemma:strong}
Let $\X \subseteq \R^n$ be an open convex set, and let $f:\R^n \to \R$ be $\mu$-strongly convex on $\X.$ Define $x^* = \argmin_{x \in \R^n} f(x)$, and assume that $x^*$ exists and $x^* \in \X.$ Then for all $x \in \X$ we have that
\[ \|x-x^*\|_2^2 \le \frac{2(f(x)-f(x^*))}{\mu}. \]
\end{lemma}
\section{Proof of \cref{thm:mainopt}}
\label{sec:opt}
In this section we show \cref{thm:mainopt} through a series of reductions, where we work with a more general space of regression problems. We first show that smoothed quadratic and $\ell_p$ regression problems with bounded entries may be reduced to solving smoothed quadratic and $\ell_p^p$ regression, so that the $\ell_p$ norm piece is instead raised to the $p$ power, largely following the approach of \cite{LS19} Section B.7.
\begin{lemma}
\label{lemma:reduce1}
Let $A \in \R^{n\times m}$ be a matrix and $d \in \R^n$ a vector, all with entries bounded by $2^{\poly(\log m)}.$ Assume that all nonzero singular values of $A$ are between $2^{-\poly(\log m)}$ and $2^{\poly(\log m)}$. For $1 \le i \le m$ let $0 \le a_i \le 2^{\poly(\log m)}$ be constants and $q_i:\R\to\R$ be functions such that $|q_i(0)|, |q_i'(0)| \le 2^{\poly(\log m)}$ and $a_i/4 \le q_i''(x) \le 4a_i$ for all $x\in\R$. For $1 \le i \le m$ let $0 \le b_i \le 2^{\poly(\log m)}$ be constants and $h_i:\R\to\R$ be functions such that $h_i(0) = h_i'(0) = 0$ and $b_i/4 \le h_i''(x) \le 4b_i$ for all $x\in\R$. For an even integer $p \le \log m$ define \[ \val(x) \defeq \sum_{i=1}^m q_i(x_i) + \left(\sum_{i=1}^m h_i(x_i)^p \right)^{1/p} \text{ and } OPT \defeq \min_{Ax=d} \val(x). \] We can compute an $x'$ with $Ax'=d$ and $\val(x') \le OPT + \err$ in $\O(1)$ oracle calls which for $0 \le W \le 2^{\poly(\log m)}$ and
\begin{equation} \val_{p,W} \defeq \sum_{i=1}^m q_i(x_i) + W\sum_{i=1}^m h_i(x_i)^p \text{ and } OPT_{p,W} \defeq \min_{Ax=d} \val_{p,W}(x) \label{eq:solvenext} \end{equation} computes a $x'$ with $Ax'=d$ and $\val_{p,W}(x') \le OPT_{p,W} + \err.$
\end{lemma}
\begin{proof}
We will apply Lemma B.3 from \cite{LS19}. Thus we must define the functions $f, g, h$ and choose constants $C_0,\chi,T_1,T_2,\mu_f,L_g,L_h,Z_1,Z_2,H_1,H_2,\eps_1$ satisfying its constraints. For simplicity, we assume that $A$ has rank $n$, so that $AA^T$ is invertible. In the case where $A=B^T$, a graph incidence matrix, the nullspace is simply the $1$ vector, and the analysis can proceed similarly. We pick $C_0=1$ and $\chi = \{x\in\R^n : \|x\|_\infty < 2^{\poly(\log m)}\},$ which is valid because all derivatives of $q_i,h_i$ and condition number of $A$ are bounded by $2^{\poly(\log m)}$.

\paragraph{Regularizing the objective.} Set $\nu = 2^{-\poly(\log m)}$ to be some sufficiently small parameter, and replace each $q_i(x) \to q_i(x) + \nu x^2$. Clearly, for all $x \in \chi$, the value of the objective is affected by at most $\nu\|x\|_2^2 = 2^{-\poly(\log m)}$ for sufficiently small $\nu$. From this point forwards, we assume that $q_i''(x) \ge 2\nu$ for all $x$. We also assume that the demand $d$ has all components at least $\nu$, as changing $d \to d+\nu1$ affects the objective value by at most $2^{-\poly(\log m)}.$

\paragraph{Reduction to unconstrained problem and choice of $f,g,h$.} We first reduce to the unconstrained case by removing the constraint $Ax=d$. Define $x_0 = A^T(AA^T)^{-1} d$, so that $Ax_0 = d$, $P\in \R^{m\times (m-n)}$ be an isomorphism onto the nullspace of $A$, which may be computed by inverting an arbitrary $n \times n$ minor of $A$. Specifically, if $A = \begin{bmatrix} X & Y \end{bmatrix}$ where $X \in \R^{n\times(m-n)}$ and $Y\in \R^{n\times n}$ is invertible, we set $P = \begin{bmatrix} I_{m-n} \\ Y^{-1}X \end{bmatrix}.$ In the case $A=B^T$, we may take $P$ to be determined by the case where $Y$ corresponds to a tree. Thus, we may replace the condition $Ax=d$ with $x=Py+x_0$ for some $y\in\R^{m-n}$. Our choice of $f,g,h$ are there
\[ f(y) \defeq \sum_{i=1}^m q_i([Py+x_0]_i) \enspace \text{ and } \enspace g(y) \defeq \left(\sum_{i=1}^m h_i([Py+x_0]_i)^p \right)^{1/p} \enspace \text{ and } \enspace h(x) = x^p. \]

\paragraph{Choice of remaining parameters.} As we have regularized each $q_i(x)$, we have that $\g^2 f(x) \se 2P^T\nu P \se 2^{-\poly(\log m)}I$, as $\nu \ge 2^{-\poly(\log m)}$ and $P^TP \se 2^{-\poly(\log m)}I$ by the condition number bound on $A$. Thus, we may set $\mu_f = 2^{-\poly(\log m)}$ by \cref{lemma:equiv}. As all entries of $A$ are bounded by $2^{\poly(\log m)}$ and all entries of $d$ are at least $\nu$ in absolute value by our reduction, all $x$ with $Ax=d$ are at least $2^{-\poly(\log m)}$ in some coordinate. Thus $f(y),g(y) \ge 2^{-\poly(\log m)}$, so we may set $T = 2^{-\poly(\log m)}.$ By our choice of $\chi$, $f(y),g(y) \le 2^{\poly(\log m)}$ for all $y\in\chi$, so we set $T_2 = 2^{\poly(\log m)}.$

We may set $L_g = 2^{\poly(\log m)}$, and for $p\le\log m$ and our choice $h(x) = x^p$ and $T_1,T_2$ we may set $L_h = 2^{\poly(\log m)}.$ We set $H_1 = h(T_1) \ge 2^{-\poly(\log m)}$ and $H_2 = h(T_2) \le 2^{\poly(\log m)}.$ Finally, we set $\eps_1 = 2^{-\poly(\log m)}.$ We set $Z_1 = 0$ and $Z_2 = \frac{C_0}{h'(T_1)} \le 2^{\poly(\log m)}.$ These parameters satisfy all desired properties, and
\[ \log \max\{ L_g,L_h,Z_2,H_2\} = \O(1) \enspace \text{ and } \enspace \log \min\{ \mu_f,\eps_1\} = -\O(1). \] This way, Lemma B.3 line (21) tells us it suffices to make $O\left(\log \frac{H_2Z_2L_gL_h}{\mu_f\eps_1}\right) = \O(1)$ oracle calls with accuracy parameter $\frac{\mu_f\eps_1^2}{100Z_2^2L_g^4L_h^2} \ge 2^{-\poly(\log m)}$, as desired.

\paragraph{Finishing the proof.} We now argue that the oracle described in \cref{lemma:reduce1} satisfies lines (19) and (20) in Lemma B.3 of \cite{LS19}. Line (19) follows by definition, and line (20) follows by line (19), $2^{\poly(\log m)}$ smoothness of the objective, and \cref{lemma:smoothgrad}. Thus, applying Lemma B.3 gives a $y$ with \[ \|\g f(y) + \g g(y)\|_2 \le \eps_1. \] Let $y^* = \argmin_y f(y) + g(y).$ By convexity we have that
\begin{align*} (f(y) + g(y)) - (f(y^*) + g(y^*)) &\le \left(\g f(y) + \g g(y)\right)^T(y-y^*) \\ &\le \|\g f(y) + \g g(y)\|_2\|y-y^*\|_2 \le 2^{-\poly(\log m)}, \end{align*} by our choice $\eps = 2^{-\poly(\log m)}$ that $\|y-y^*\|_2 \le 2^{\poly(\log m)}$ from our choice of $\chi$.
\end{proof}
We now show that objectives as in \cref{eq:solvenext} may be iteratively refined.
\begin{lemma}
\label{lemma:reduce2}
Let $A \in \R^{n\times m}$ be a matrix and $d \in \R^n$ a vector, all with entries bounded by $2^{\poly(\log m)}.$ Assume that all nonzero singular values of $A$ are between $2^{-\poly(\log m)}$ and $2^{\poly(\log m)}$. For $1 \le i \le m$ let $0 \le a_i \le 2^{\poly(\log m)}$ be constants and $q_i:\R\to\R$ be functions such that $|q_i(0)|, |q_i'(0)| \le 2^{\poly(\log m)}$ and $a_i/4 \le q_i''(x) \le 4a_i$ for all $x\in\R$. For $1 \le i \le m$ let $0 \le b_i \le 2^{\poly(\log m)}$ be constants and $h_i:\R\to\R$ be functions such that $h_i(0) = h_i'(0) = 0$ and $b_i/4 \le h_i''(x) \le 4b_i$ for all $x\in\R$. For an even integer $p \le \log m$ define $\val_{p,W}(x)$ and $OPT_{p,W}$ as in \cref{eq:solvenext}. We can compute a $x'$ with $Ax'=d$ and $\val_{p,W}(x') \le OPT_{p,W} + \err$ with one call to a solver for $(A^TA)^\dagger d$ and $\O(2^{22p})$ oracle calls which for $g \in \R^m, r \in \R_{\ge0}^m$, all entries bounded by $2^{\poly(\log m)}$, and \[ \val_{g,r,b}(x) \defeq \sum_{i=1}^m g_ix_i + \left(\sum_{i=1}^m r_ix_i^2\right) + \sum_{i=1}^m b_i^px_i^{2p} \enspace \text{ and } \enspace OPT_{g,r,b} = \min_{Ax=0} \val_{g,r,b}(x) \] computes an $x'$ with $Ax'=0$ and $\val_{g,r,b}(x') \le OPT_{g,r,b} + \err.$
\end{lemma}
\begin{proof}
We use \cref{algo:reduce2}, which reduces \cref{eq:solvenext} to calls to $\OracleToP$.

\begin{algorithm}[h]
\caption{$\ReduceToP(A,d,q,h,a,b)$. Takes matrix $A\in\R^{n\times m}$, vector $d$, constants $a_i,b_i$, and functions $q_i,h_i$ for $1\le i \le m$. Computes $x$ with $Ax=d$ and $\val_{p,W}(x) \le OPT_{p,W} + \err$ with $\O(2^{22p})$ calls to $\OracleToP(A,g,r,b)$, which computes an $x$ with $Ax=0$ and $\val_{g,r,b}(x') \le OPT_{g,r,b} + \err.$}
Initialize $x \assign A^T(AA^T)^\dagger d$. \label{line:electric} \\
\For{$1 \le t \le \O(2^{22p})$}{\label{line:iteratefor}
	Initialize $r$. \\
	\For{$1 \le i \le m$}{
		$r_i = a_i + b_i^px_i^{2p-2}.$ \label{line:pickr} \\
	}
	$g \assign \g \val_{p,W}(x)$. \label{line:pickg} \\
	$\tD \assign \OracleToP(A,g,2^{-16p}r,2^{-16}b).$ \\
	$x \assign x + 2^{-22p}\tD.$ \label{line:iteration} \\
}
Return $x$.
\label{algo:reduce2}
\end{algorithm}

Note that all entries of $x = A^T(AA^T)^\dagger d$ are bounded by $2^{\poly(\log m)}$ because all nonzero singular values of $A$ are between $2^{-\poly(\log m)}$ and $2^{\poly(\log m)}$. Therefore, $\val_{p,W}(x) \le 2^{\poly(\log m)}$.

We show that iteration in line \ref{line:iteration} decreases the value of the objective on $x$ multiplicatively towards $OPT$.
Applying \cref{lemma:iterh} \cref{eq:2iter} to $q_i$ with $c_1 = a_i/4$ and $c_2 = 4a_i$ gives
\begin{equation} \frac18a_i\D_i^2 \le q_i(x_i+\D_i)-q_i(x_i)-q_i'(x_i)\D_i \le 2a_i\D_i^2. \label{eq:add1} \end{equation}
Applying \cref{lemma:iterh} \cref{eq:piter} to $h_i$ with $c_1 = b_i/4$ and $c_2 = 4b_i$ gives
\begin{equation} 2^{-16p}b_i^p\left(x_i^{2p-2}\D_i^2 + \D_i^{2p} \right) \le h_i(x_i+\D_i)^p-h_i(x_i)^p-p\cdot h_i(x_i)^{p-1}h_i'(x_i) \le 2^{6p}b_i^p\left(x_i^{2p-2}\D_i^2 + \D_i^{2p} \right). \label{eq:add2} \end{equation}
Adding \cref{eq:add1,eq:add2} for all $i$ and using our choice of $g,r$ in line \ref{line:pickg}, \ref{line:pickr} gives that for $\D$
\begin{equation} 2^{-16p}\left(\sum_{i=1}^m r_i\D_i^2 + \sum_{i=1}^m b_i^p\D_i^{2p}\right) \le \val_{p,W}(x+\D)-\val_{p,W}(x)-g^T\D \le 2^{6p}\left(\sum_{i=1}^m r_i\D_i^2 + \sum_{i=1}^m b_i^p\D_i^{2p}\right) \label{eq:add3} \end{equation}
Applying \cref{lemma:refineprogress} to \cref{eq:add3} gives us that
\[ \val_{p,W}(x+2^{-22p}\D)-OPT_{p,W} \le (1-2^{-22p})\left(\val_{p,W}(x)-OPT_{p,W}\right) + \err. \] As $\val_{p,W}(x) \le 2^{\poly(\log m)}$ initially, performing this iteration $\O(2^{22p})$ times as in line \ref{line:iteratefor} results in $\val_{p,W}(x) \le OPT_{p,W} + \err$ at the end, as desired.
\end{proof}
Now, the proof of \cref{thm:mainopt} follows from verifying the conditions of \cref{lemma:reduce1,lemma:reduce2} and applying \cref{thm:smoothflow}.
\begin{proof}[Proof of \cref{thm:mainopt}]
We check the conditions of \cref{lemma:reduce1,lemma:reduce2}. By Lemma D.1 of \cite{LS19} all resistances and residual capacities are polynomially bounded, hence all quantities encountered during the algorithm are $2^{\poly(\log m)}$. For $A = B^T$, we have that $AA^T = B^TB$ is a graph Laplacian, and hence has polynomially bounded singular values. We may compute an initial point $x = B(B^TB)^\dagger d$ as in line \ref{line:electric} by computing an electric flow. All remaining conditions follow directly for the choice $b_i=1$.

We now analyze the runtime. We set $b_i=1$ to apply \cref{lemma:reduce1,lemma:reduce2} to \cref{thm:mainopt}. \cref{thm:smoothflow} allows us to solve the objectives desired by $\OracleToP(B^T,g,r,b)$ in $m^{1+o(1)}$ time for $b = 1$ and $p = 2\left\lceil\sqrt{\log m}\right\rceil$. Additionally, $\O(2^{22p}) = m^{o(1)}$ for $p = 2\left\lceil\sqrt{\log m}\right\rceil$, hence the total runtime is $m^{1+o(1)}$ as desired.
\end{proof}

\end{document}